\documentclass[pra,aps,preprint,showpacs,preprintnumbers,
amsmath,amssymb,floatfix]{revtex4-1}
\usepackage{enumerate}
\usepackage{graphicx}
\usepackage{amsmath, amsthm, amssymb}
\newtheorem{theorem}{Theorem}

\newtheorem{lemma}{Lemma}

\newtheorem{definition}{Definition}
\newtheorem{proposition}{Proposition}

\theoremstyle{remark}
\newtheorem*{remark}{Remark}

\usepackage{enumerate}
\usepackage{graphicx}
\usepackage{amsmath, amssymb}
\usepackage{graphics}
\usepackage{bm}
\begin{document}
\newcommand{\real}{\textrm{Re}\:}
\newcommand{\QQW}{\mathfrak{X}}
\newcommand{\sto}{\stackrel{s}{\to}}
\newcommand{\Tr}{\textrm{Tr}\:}
\newcommand{\Ran}{\textnormal{Ran}}
\newcommand{\supp}{\textrm{supp}\:}
\newcommand{\evs}{\textnormal{evs}\:}
\newcommand{\Ker}{\textnormal{Ker}\:}
\newcommand{\sign}{\textnormal{sign}\:}
\newcommand{\diag}{\textnormal{diag}}
\newcommand{\wto}{\stackrel{w}{\to}}
\newcommand{\ssto}{\stackrel{s}{\to}}
\newcommand{\epstr}{{\tilde \epsilon}_0}
\newcommand{\beps}{\pmb{\epsilon}}
\newcounter{foo}
\newcommand{\sslim}{\textnormal{s--}\lim}
\newcommand{\wlim}{\textnormal{w--}\lim}
\providecommand{\norm}[1]{\lVert#1\rVert}
\providecommand{\abs}[1]{\lvert#1\rvert}
\providecommand{\absbm}[1]{\pmb{|}\mspace{1mu}#1\mspace{1mu}\pmb{|}}
\newcommand{\br}{{\bf r}}
\newcommand{\by}{{\bf y}}
\newcommand{\bq}{{\bf q}}
\newcommand{\bp}{{\bf p}}

\title{Proof of the Super Efimov Effect }

\author{Dmitry K. Gridnev}
\affiliation{FIAS, Ruth-Moufang-Stra{\ss}e 1, D--60438 Frankfurt am Main,
Germany}
\altaffiliation[On leave from:  ]{ Institute of Physics, St. Petersburg State
University, Ulyanovskaya 1, 198504 Russia}

\begin{abstract}
We consider the system of 3 nonrelativistic spinless fermions in two dimensions,
which interact through spherically-symmetric pair interactions. 
Recently a claim has been made for the existence of the so-called super Efimov
effect [Y. Nishida {\it et al.}, Phys. Rev. Lett. 110, 235301 (2013)]. Namely,
if 
the interactions in the system are fine-tuned to a p-wave resonance, an infinite
number of bound states appears, whose negative energies are scaled according to 
the double exponential law. We present the 
mathematical proof that such system indeed has an infinite number of bound
levels. We also prove that 
$\lim_{E \to 0} |\ln|\ln E||^{-1} N(E) = 8/(3\pi) $, where $N(E)$ is the number
of bound states with the energy less than $-E <0$. 
The value of this limit is equal exactly to the value derived in [Y. Nishida
{\it et al.}] using 
renormalization group approach. Our proof resolves a recent controversy about
the validity of results in [Y. Nishida {\it et al.}]. 
\end{abstract}

\maketitle

\section{Introduction}\label{sec:1}

Efimov effect first discovered by V. Efimov in \cite{1} is one of the most
intriguing phenomena in physics. This effect appears in 3-body systems in
3-dimensional space, which interact through 
short-range pair-potentials. It is always possible to tune the couplings of the
interactions in such a way that none of the particle pairs has a negative energy
bound state, but at least two 
pairs have a resonance at zero energy. In this case the 3-body system exhibits
an infinite sequence of bound levels, where the energy of the $n$-th level
decreases exponentially with $n$. 
The rigorous proof of this effect in \cite{2,3} is a highlight of mathematical
physics. Suppose that three particles are identical, the pair interaction is  
tuned to the zero energy resonance, and let $N_E$ be the total number of 3-body
bound states with the 
energy less than $-E <0$. Then $\lim_{E \to 0} |\ln E|^{-1} N_E = s_0/(2\pi)$,
where $s_0$ is the root of the known transcendental 
equation expressed in elementary functions \cite{3}.

Relatively recently the authors in \cite{4} considered the system of 3 spinless
fermions in flatland using field-theoretical methods. 
The spherically symmetric pair interaction of fermions was tuned in such 
a way that pairs of fermions had no negative spectrum but were at the coupling
constant threshold \cite{5,6}, so that a negligible increase of the coupling 
constant would result in the formation of an antisymmetric 2-body bound state with
negative energy. In this case one says that the interactions 
are tuned to the zero energy p-wave resonance. In \cite{4} the authors came to
the conclusion that such system has two 
infinite series of bound states each corresponding to the orbital 
angular momentum $L = \pm 1$. The energies of these bound states $E_n$ for large
$n$ were predicted to approach the form 
$E_n \sim -\exp\bigl(-2e^{\frac{3\pi n}4 + \theta}\bigr)$, 
where $\theta$ is a constant defined modulo $3\pi/4$. The authors termed this
phenomenon "super Efimov effect". 
If $N_E$ is the total number of 3-body bound states with the energy less than
$-E <0$
(for all values of the angular momentum) then the results in \cite{4} predict
that 
\begin{equation}\label{1}
\lim_{E\to 0} |\ln|\ln E|^{-1} N_E = 8/(3\pi) . 
\end{equation}
There are two interesting features about the super Efimov effect.  
First, it turns out that the system of 3 spinless fermions in two dimensions may
have an infinite number of bound states, though 
the same system in 3 dimensions has at most a finite number of levels with
negative energy \cite{7}. Secondly, the energy of the $n$-th level goes
extremely fast to zero with increasing $n$. This is reflected in the double
logarithm in (\ref{1}) and differs from the Efimov effect of 3 bosons in
3-dimensional space, where a single logarithm enters the similar formula
\cite{3}. 

Recently in the physics literature there were raised doubts about whether the
super Efimov effect 
is real \cite{hammer,8,9}. In \cite{8} it was claimed that the sequence of
levels with double exponential scaling does not exist and instead 
there emerges another infinite sequence of levels, which approaches the scaling
law $E_n \sim -\exp (n^2 \pi^2 /2Y) $  with $Y >0$ being 
a non-universal constant. The findings in \cite{8} are in contradiction with Eq.~(\ref{1}). In
\cite{9} the authors 
observed the super Efimov effect in the lowest order of the hyperspherical
expansion. 
However, the value of the limit in (\ref{1}) was found to be $2(16/9 -1/4)^{-1/2}$; the 
inclusion of higher order effects could not provide definitive conclusions on
whether the infinite sequence of levels exists. 
In the present paper we shall provide a rigorous mathematical proof of (\ref{1}). 
Hence, we demonstrate that the super Efimov effect is indeed real and the
constant on the rhs of (\ref{1}) coincides exactly with the one predicted in \cite{4}. 

The basic idea behind the proof of (\ref{1}) stems from \cite{2}, namely, one
uses symmetrized Faddeev equations and the Birman-Schwinger 
principle \cite{birman,5,6,10} for counting eigenvalues. 
Like in \cite{2} we reduce the problem to counting the eigenvalues in the
interval $(0,\infty)$ of an integral operator, 
which depends on the energy. 
Let us explain, however, the major difference. In \cite{2} when the energy
approached zero this integral operator approached (in the strong sense) a
bounded 
integral operator, which had a nonempty essential spectrum in the interval
$(1,\infty)$. In the 2-dimensional case a similar  integral operator maintains
discrete 
spectrum but its norm goes to infinity when the energy goes to zero. The control
of appearing error terms becomes challenging  because their norm diverges as well. 

We shall use the following notations. An abstract Hilbert space $\mathcal{H}$ is
assumed to be separable, $\mathcal{C}(\mathcal{H})$ denotes the ideal of all
compact operators on $\mathcal{H}$. 
For a self-adjoint operator $A \in \mathcal{C}(\mathcal{H})$ we denote by
$\lambda_1 (A), \lambda_2 (A), \ldots $ its non-negative eigenvalues (counting
multiplicities) 
in descending 
order; if this sequence terminates at $n_0$ we set $\lambda_{n_0+1} (A) =
\lambda_{n_0+2} (A) = \cdots =0$. For a self-adjoint operator $A$ on
$\mathcal{H}$ we shall denote by  $D(A)$, $\sigma(A)$ and $\sigma_{ess} (A)$ the
domain, the spectrum, and 
the essential spectrum of $A$ respectively \cite{11}. $A \geq 0$ means that $(f,
Af) \geq 0$ for all $f \in D(A)$, while 
$A \ngeq 0$ means that there exists $f_0 \in D(A)$ such that $(f_0, Af_0) < 0$. 
$\mathfrak{n}(A, a)$ is the number of eigenvalues of $A$ (counting
multiplicities) that are larger than $a >0$. 
By $\mu_n (A)$ we denote singular values of $A \in \mathcal{C}(\mathcal{H})$ listed in descending order 
\cite{17}. Similarly, 
$\mathfrak{n}_\mu (A, a) = \mathfrak{n} (|A|, a)$ is the number of singular
values of $A\in \mathcal{C}(\mathcal{H})$ that are larger than $a >0$.  
$\|A\|_{HS}$ is the Hilbert-Schmidt norm of an operator $A $.  For an interval
$\Omega \subset \mathbb{R}$ the function 
$\chi_\Omega : \mathbb{R} \to \mathbb{R}$ is such that $\chi_\Omega (x) = 1$ if
$x\in \Omega$ and $\chi_\Omega (x) = 0$ otherwise. 
$\diag \{a_1, a_2 , a_3\}$ denotes a $3\times 3$ matrix 
with the diagonal entries $a_1, a_2 , a_3$ and zero off-diagonal elements. 

\section{Main Result}\label{sec:2}

We shall consider 3 spinless fermions in $\mathbb{R}^2$ that interact through
$v(|r_i -r_k |) \leq 0$, where $r_i$ are particle position vectors. For pair
interactions 
we assume that $v$ is a Borel function,  $|v(x)| \leq \alpha_1 e^{-\alpha_2
|x|}$ with $\alpha_{1,2} >0$ being constants. 
Regarding the fermion's mass $m$ we shall use the units, where $\hbar^2 /m =1$.
The Hamiltonian of this system reads
\begin{equation}\label{3}
 H = H_0 + \sum_{1 \leq i<k\leq 3} v(|r_i -r_k |) , 
\end{equation}
where $H_0$ is the kinetic energy operator with the removed center of mass
motion. Due to the Pauli principle $H$ should be considered on an
antisymmetrized space, 
which is constructed below. 
For $k=1,2,3$ let $x_k, y_k \in \mathbb{R}^2$ be three sets of Jacobi
coordinates, which are shown in Fig.~\ref{fig:1} 
\begin{gather}
 x_k = r_i - r_j \\
y_k = \frac{2}{\sqrt 3} \left[r_k - \frac{r_i +r_j}2 \right] , 
\end{gather}
where $(k,i,j)$ is an odd permutation of $(1,2,3)$. The scalings are chosen so
that in all coordinate sets $H_0 = - \Delta_{x_{k}}-\Delta_{y_{k}}$. The
coordinate 
sets  are connected through the orthogonal linear transformation 
\begin{equation}\label{6}
 \begin{pmatrix}
  x_i \\
y_i
 \end{pmatrix}
=
 \begin{pmatrix}
  -\frac 12 &\frac{\sqrt 3}2 \\
  -\frac{\sqrt 3}2  &-\frac 12
 \end{pmatrix}
 \begin{pmatrix}
  x_k \\
y_k
 \end{pmatrix} , 
\end{equation}
where $(i,k,j)$ is an odd permutation of $(1,2,3)$. 
Let us write Jacobi coordinates as functions of particle position vectors, that
is 
$x_k = x_1 (r_1 , r_2, r_3)$ and $y_k = y_1 (r_1 , r_2, r_3)$. And let $p =
\bigl(p(1), p(2), p(3)\bigr)$ be a permutation of $(1,2,3)$. 
Then by definition $p (x_1) = x_1 (r_{p(1)} , r_{p(2)}, r_{p(3)})$ and $p (y_1)
= x_1 (r_{p(1)} , r_{p(2)}, r_{p(3)})$. 
We define the action of the permutation operator $p$ on $L^2 (\mathbb{R}^4)$ as
$p f (x_1 , y_1) = f(p(x_1), p(y_1))$. Now we define the subspace 
of antisymmetric square-integrable functions as $L^2_A (\mathbb{R}^4) =
\bigl\{\psi \bigl|\psi \in L^2 (\mathbb{R}^4)\textnormal{ and }p\psi =
(-1)^{\pi(p)}\psi\bigr\}$, whereby 
$\pi(p)$ is the parity of the permutation $p$. By standard results \cite{11,12}
the Hamiltonian $H$ is self-adjoint on $L^2_A (\mathbb{R}^4)$ with the domain
$D(H) = L^2_A (\mathbb{R}^4)\cap \mathcal{H}^2 (\mathbb{R}^4)$, where 
$\mathcal{H}^2 (\mathbb{R}^4)$ is the corresponding Sobolev space \cite{12,13}. 

The subsystem of 2 fermions is described by the Hamiltonian $h(1)$, where 
\begin{equation}\label{7}
h(\lambda)   = -\Delta_x + \lambda v(|x|) 
\end{equation}
and $\lambda > 0$ is a coupling constant. The Hamiltonian (\ref{7}) acts on the
subspace $L_A^2 (\mathbb{R}^2)$, where 
$L_A^2 (\mathbb{R}^2) = \bigl\{\phi \bigl|\phi \in L^2
(\mathbb{R}^2)\textnormal{ and }\phi(x) = -\phi(-x)\bigr\}$. 
$h(\lambda)$ is self-adjoint on $L_A^2 (\mathbb{R}^2) $ with the domain 
$D(h) = L^2_A (\mathbb{R}^2)\cap \mathcal{H}^2 (\mathbb{R}^2)$. We shall say
that the interaction $v(x)$ is \textit{tuned to the p-wave zero energy
resonance} if $h(1) \geq 0$ and 
$h(1+\varepsilon) \ngeq 0$ for all $\varepsilon >0$. 

\begin{figure}
\includegraphics[height=0.16\textheight]{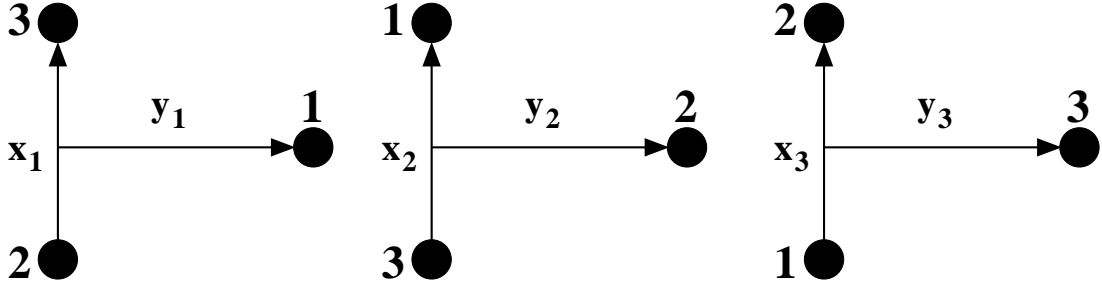}
\caption{Three sets of Jacobi coordinates. The picture shows only directions of
the vectors, the scales are chosen in order to ensure that $H_0 =
-\Delta_{x_k}-\Delta_{y_k}$ for $k=1,2,3$.}
\label{fig:1}
\end{figure}

Let $N_z (H)$ denote the number of bound states of $H$, whose energy is less
than $-z^2$. Our aim in this paper is to prove 
the following 
\begin{theorem}\label{th:1} Suppose that the interactions in (\ref{3}) are tuned
to the zero energy p-wave resonance. Then 
$\lim_{z \to 0} |\ln|\ln z^2 | |^{-1} N_z (H) = 8/(3\pi). $
\end{theorem}
\begin{remark}
 Theorem 1 provides a firm mathematical footing for the super Efimov effect. We do not prove it here, but one can
show that 
$\lim_{z \to 0} |\ln|\ln z^2 | |^{-1}  N^\pm_z (H) = 4/(3\pi)$, where $N^\pm_z
(H)$ is the number of bound states of $H$, which 
have the energy less than $-z^2$ and angular momentum $\pm 1$ respectively. This
agrees with the results in \cite{4}.  
\end{remark}

From now on we shall always assume that the interaction of 2 spinless fermions
is tuned to the zero energy p-wave resonance. 
Consider the integral operator on $L^2_A (\mathbb{R}^2)$
\begin{equation}\label{8}
 k(z) := |v|^{\frac 12} (-\Delta_x + z^2)^{-1} |v|^{\frac 12} , 
\end{equation}
which is called the Birman-Schwinger (BS) operator. Its integral kernel has the
form (eq. (7.2) in \cite{5}) 
\begin{equation}
 k(x,y) = (2\pi)^{-1} |v(x)|^{\frac 12} K_0 (z|x-y|) |v(y)|^{\frac 12} , 
\end{equation}
where 
\begin{gather}
 K_0 (t) = -\tilde I_0 (t) \ln \frac t2 + \sum_{m=0}^\infty \frac{t^{2m}}{2^{2m}
(m!)^2} \psi (m+1) , \\
\tilde I_0 (t) = \sum_{l=1}^\infty \bigl[(l!)\Gamma (l+1)\bigr]^{-1} 2^{-2l}
t^{2l} , \label{11}\\
\psi (j) = 1 + \frac 12 + \cdots \cdots \frac 1j -1 - C
\end{gather}
with $C$ being Euler's constant. Note that contrary to \cite{5} the summation in
(\ref{11}) starts from $l=1$ because the term produced in $k(z)$ by $l=0$ is
identically zero on 
$L^2_A (\mathbb{R}^2)$ (this term, which is responsible for the projection
operator term in (7.3) in \cite{5}, is absent in our case). Thus we have
\cite{5} 
$k(z) = [\sum_{k=0}^\infty A_k z^{2k}]z^2\ln z + [\sum_{k=0}^\infty B_k
z^{2k}]$, where the series in square brackets sum up to 
entire analytic operator functions and the coefficients 
$A_k, B_k$ are Hilbert-Schmidt operators. The operator $k(0)$ is compact and in
the vicinity of $z=0$ the operator $k(z)$ is compact as well. 
Because the interaction is tuned to the p-wave zero energy resonance from the BS
principle (see Theorem 9 in \cite{10}) we infer that $\|k(0)\| = 1$. 
By standard results in quantum mechanics the ground state of $h(\lambda)$ for
$\lambda >1$ is doubly degenerate with the angular momentum 
$l=\pm 1$. By the BS principle \cite{10} it follows immediately that $\|k(0)\| =
1$ is an eigenvalue of $k(0)$ with multiplicity 2. 
Due to spherical symmetry the largest eigenvalue of $k(z)$ for $z>0$ is also
doubly degenerate. By the analysis in \cite{5} in the vicinity of $z = 0$ one
has 
\begin{equation}\label{13}
k(z) \varphi_\pm (z) = \mu(z) \varphi_\pm (z) , 
\end{equation}
where $z \geq 0$, $\mu(z)= \sup \sigma (k(z))$,  
$\|\varphi_\pm (z)\| =1$, $\mu(0)=1$. 
The orthogonal eigenvectors $\varphi_\pm (z)$ are defined for all $z\geq 0$ and
are eigenfunctions of the angular momentum operator with the eigenvalues $l = \pm 1$
respectively.  
By standard results in perturbation theory we have 
\begin{equation}\label{14}
 \varphi_\pm (z) = \eta_\pm + \mathcal{O}(z^2\ln z), 
\end{equation}
where $\eta_\pm \equiv \varphi_\pm (0)$. Due to the spherical symmetry of the
potential $\eta_\pm (x) = \eta_0 (|x|) e^{\pm i \varphi_x}$, where 
$|x|, \varphi_x$ are polar coordinates. By perturbation theory \cite{5} $\mu(z)$
has a convergent expansion in the vicinity of $z=0$ given by the series 
$\mu(z) = \sum_{n,m \geq 0}^\infty c_{nm} (z^2 \ln z)^n z^{2m}$, where $c_{nm}
\in \mathbb{R}$.  The leading terms of perturbation series are given by 
the expression 
\begin{equation}\label{15}
 \mu(z) = 1 + \frac {\pi}2 c_0^2 z^2 \ln z + \mathcal{O}(z^2), 
\end{equation}
where 
\begin{gather}
 c_0^2 = - \frac 1{4\pi^2} \int \int |v (|x|)|^\frac 12 |v(|y|)|^\frac 12 
|x-y|^2 \eta_+^* (x) \eta_+ (y) d^2 x d^2 y \nonumber\\
= \left[\int_0^\infty s^2 \eta_0 (s) |v (s)|^\frac 12 \right]^2 \label{17}
\end{gather}
(see also the text below eq. (7.13) in \cite{5}). Note that due to $k(z) >
k(z')$ for $z' > z >0$  the function $\mu (z)$ is monotone decreasing on $[0,
\infty)$ and 
\begin{equation}
\sup_{z> 0} \left\| \bigl[ 1 - \varphi_+ (z) (\varphi_+ (z) , \cdot) - \varphi_-
(z) (\varphi_- (z) , \cdot) \bigr] k(z)\right\| <  1 . \label{09.81}
\end{equation}

Now we consider the 3-body problem. We denote $v_\alpha := v(|r_\beta -
r_\gamma|)$, where $(\alpha , \beta , \gamma)$ is any permutation of the numbers
$(1,2,3)$. 
Let us introduce the linear subspace $\mathcal{H}_A \subset L^2 (\mathbb{R}^4)
\oplus L^2 (\mathbb{R}^4)  \oplus L^2 (\mathbb{R}^4) $, where each vector 
$(\phi_1, \phi_2, \phi_3) \in \mathcal{H}_A$ satisfies the antisymmetry
requirements listed in Table~\ref{tab:1}. Each operator $p_{ik}$ in
Table~\ref{tab:1} permutes 
spatial coordinates of particles $i,k$. 
Let us consider the operator $M(z)$ on $\mathcal{H}_A$, whose matrix entries are
the following 
operators 
\begin{equation}\label{09.100}
 M_{\alpha \beta } (z) := |v_\alpha|^{\frac 12} (H_0 + z^2)^{-1} |v_\beta|^{\frac
12} . 
\end{equation}
For each set of Jacobi coordinates in Fig.~\ref{fig:1} we introduce the Fourier 
transform $\mathcal{F}_k$, which acts on $f(x_k, y_k) \in L^1 (\mathbb{R}^4)$ as
follows 
\begin{equation}
 \hat f(p_k, q_k) := \frac 1{(2\pi)^2} \int d^2 x_k d^2y_k e^{-i(x_k \cdot p_k +
y_k \cdot q_k)} f(x_k, y_k) . 
\end{equation}
For any interval $\Omega \subset \mathbb{R}$ let us define the cutoff operator
on $\mathcal{H}_A$ 
\begin{equation}\label{cutoff25}
 \mathfrak{X}_{\Omega} = \diag \left\{\mathcal{F}_1^{-1} \chi_{\Omega} (|q_1|)
\mathcal{F}_1, \mathcal{F}_2^{-1} \chi_{\Omega} (|q_2|) \mathcal{F}_2,
\mathcal{F}_3^{-1} \chi_{\Omega} (|q_3|) \mathcal{F}_3\right\} . 
\end{equation}
We separate the diagonal part of $M(z)$ by writing $M(z) = M' (z) + M_d (z)$,
where 
$M_d (z) := \diag \{M_{11} (z), M_{22} (z), M_{33} (z)\}$ and $M' (z) = M (z) -
M_d (z)$. The operator $M' (z)$ is compact for all $z>0$. Indeed, we can write 
\begin{gather}
 M' (z) = \mathfrak{X}_{[0, R]}  M' (z) \mathfrak{X}_{[0, R]} +
\mathfrak{X}_{(R, \infty)}  M' (z) \mathfrak{X}_{[0, R]} \nonumber\\
+ \mathfrak{X}_{[0, R]}  M' (z) \mathfrak{X}_{(R, \infty)} + \mathfrak{X}_{(R,
\infty)}  M' (z) \mathfrak{X}_{(R, \infty)} . \label{21}
\end{gather}
Since the interactions are bounded it is easy to see that the norm of each of
the last three terms is $\hbox{o}(R)$ for $R \to \infty$. Hence, 
the compactness of the lhs of (\ref{21}) follows 
from the compactness of the first term on the rhs for all $R >0$. We prove its
compactness by proving the same for each of its matrix entries considered as
operators on $L^2 (\mathbb{R}^4)$. 
The operator $\mathcal{F}_1 \mathfrak{X}_{[0, R]}   M_{12}(z)\mathfrak{X}_{[0,
R]} \mathcal{F}_2^{-1}$ has the kernel 
\begin{equation}\label{22}
 {\hat M}_{12} (p_1, q_1; p_1', q_1') = \frac 1{\pi^2} \chi_{[0,R]}
(|q_1|)\frac{\widehat{|v|^{\frac 12 }} \Bigl(p_1 + \frac{2}{\sqrt 3} q_1' +
\frac{1}{\sqrt 3} q_1\Bigr) \widehat{|v|^{\frac 12 }} \Bigl(\frac{1}{\sqrt 3} q_1'
+ \frac{2}{\sqrt 3}  q_1 - p_1'\Bigr)}{(2q_1' + q_1)^2 + 3q_1^2  + 3 z^2}\chi_{[0,R]}
(|q_1'|) , 
\end{equation}
where $\widehat{|v|^{\frac 12 }} : \mathbb{R}^2 \to \mathbb{C}$ is the Fourier
transform of $|v(|x|)|^{\frac 12 }$. It is elementary to check that the
Hilbert-Schmidt norm of the operator in (\ref{22}) is finite. 
First, we prove 
\begin{lemma}\label{lem:1}
 The following equation holds $N_z (H)= \mathfrak{n} (M(z), 1)$. 
\end{lemma}
Similar lemma has been proved in \cite{3}, however, we need to give a new proof
in view of antisymmetry restrictions. 
\begin{proof}[Proof of Lemma \ref{lem:1}]
Consider the operator 
\begin{equation}
 \mathfrak{m} (z) = \sum_\alpha \bigl( H_0 + z^2 \bigr)^{-\frac 12} |v_\alpha |
\bigl( H_0 + z^2 \bigr)^{-\frac 12}
\end{equation}
 on the space $L^2_A (\mathbb{R}^4)$. 
By the BS principle \cite{10} $ N_z (H) = \mathfrak{n} \left(\mathfrak{m} (z),
1\right)$. Let $L_\lambda$ and $\mathcal{H}_\lambda$ denote 
the eigenspaces of the operators $ \mathfrak{m} (z) $ and $   M(z)$
respectively, which correspond to the eigenvalue $\lambda > 1$. Let us first
show that 
the dimension of both eigenspaces is finite. Note that by Theorem 9 in \cite{10}
$\sigma_{ess} (\mathfrak{m} (z)) \subset (-\infty, 1]$, hence, $\textnormal{dim}
L_\lambda$ is finite. 
Due to compactness of $M'(z)$ we have 
\begin{equation}
 \sigma_{ess}(M(z)) = \sigma_{ess}(M_d (z)) \subseteq [0, \sup_\alpha
\left\||v_\alpha|^{\frac 12} (H_0 +z^2)^{-1}|v_\alpha|^{\frac
12}\right\|]\subset [0,1]
\end{equation}
and 
thus $\textnormal{dim} \mathcal{H}_\lambda$ is also finite. The operator
$B_\lambda :L_\lambda \to \mathcal{H}_\lambda$ is defined as 
$(B_\lambda \psi )_\alpha =  |v_\alpha|^{1/2} (H_0 + z^2)^{-1/2} \psi $. It is
easy to check that this operator is defined correctly, and by applying this
operator 
we infer that from $\textnormal{dim} L_\lambda \neq 0$ it follows that
$\textnormal{dim} \mathcal{H}_\lambda \neq 0$. Similarly, we  
define the operator $B'_\lambda : \mathcal{H}_\lambda\to L_\lambda$ given by 
$B'_\lambda \phi =  (H_0 + z^2)^{-1/2} \sum_\beta |v_\beta|^{1/2} \phi_\beta$,
which is also well-defined. Applying 
this operator we find that $\textnormal{dim} \mathcal{H}_\lambda \neq 0
\Longleftrightarrow \textnormal{dim} L_\lambda \neq 0$.  
Since $\lambda^{-1} B'_\lambda B_\lambda = 1$ we get that 
$\textnormal{dim} L_\lambda = \textnormal{dim} \mathcal{H}_\lambda$. Therefore, 
$ N_z (H) = \mathfrak{n} \left(\mathfrak{m} (z), 1\right) = \mathfrak{n} \left(M
(z), 1\right)$. 
\end{proof}
By Lemma~\ref{lem:1} and the BS principle \cite{10} 
\begin{equation}\label{25}
 N_z (H) = \mathfrak{n} \left( \mathcal{A} (z), 1\right) , 
\end{equation}
where 
\begin{equation}
 \mathcal{A} (z) = \bigl(1 - M_d (z)\bigr)^{-\frac 12}M' (z) \bigl(1 - M_d
(z)\bigr)^{-\frac 12} . 
\end{equation}

For $k=1,2,3$ let us introduce the projection operators $P_\pm^{(k)}$ on
$\mathcal{H}_A$, which act on $f(p_k, q_k)$ as follows 
\begin{equation}\label{48}
  [P_\pm^{(k)}  f ] (p_k , q_k) = \hat \eta_\pm (p_k) \int \hat \eta_\pm^*
(p'_k) f (p'_k , q_k) dp'_k , 
\end{equation}
and 
\begin{equation}
 P_\pm  = \diag \left\{\mathcal{F}_1^{-1} P_\pm^{(1)} \mathcal{F}_1,
\mathcal{F}_2^{-1} P_\pm^{(2)} \mathcal{F}_2, \mathcal{F}_3^{-1} P_\pm^{(3)}
\mathcal{F}_3\right\} . \label{09:15}
\end{equation}
In (\ref{48}) $\hat \eta_\pm$ are Fourier transformed functions $\eta_\pm$
defined in (\ref{14}). Let us fix the cut off parameter $r_\varepsilon \in (0, 1/4)$ and
define 
\begin{equation}
 \mathcal{A}_0 (z)  =   \bigl[P_+  + P_- \bigr] G (z)M' (z) G (z) \bigl[P_+  +
P_- \bigr] , 
\end{equation}
where 
\begin{gather}
 G (z) = \diag \left\{ \mathcal{F}_1^{-1}g_z (|q_1 |)\mathcal{F}_1,
\mathcal{F}_2^{-1}g_z (|q_2 |)\mathcal{F}_2 , \mathcal{F}_3^{-1}g_z (|q_3
|)\mathcal{F}_3 \right\} , \label{9.61}
\end{gather}
and $g_z : \mathbb{R}_+ \to \mathbb{R}_+$ is defined through 
\begin{equation} 
g_z (s) := 
\begin{cases}
 \left(1 - \mu \Bigl( \sqrt{s^2 +z^2}\Bigr)\right)^{-1/2} & \textnormal{if $s
\leq r_\varepsilon$}, \\
0& \textnormal{if $s > r_\varepsilon$} . 
\end{cases}
\end{equation}
By (\ref{15}) there exist $\delta, \delta' > 0$ such
that 
\begin{equation}\label{34}
 \frac{\delta'}{(s^2 +z^2)|\ln (s^2 +z^2)|} \leq g_z^2 (s) \leq
\frac{\delta}{(s^2 +z^2)|\ln (s^2 +z^2)|} 
\end{equation}
for $z, s \in (0, r_\varepsilon ]$. Besides for $r_\varepsilon \to 0$ we
have $\delta = 4/(\pi c_0^2) + \hbox{o}(r_\varepsilon)$ and $\delta' = 4/(\pi c_0^2) +
\hbox{o}(r_\varepsilon)$. 
We shall always implicitly assume that $z \in (0,r_\varepsilon]$ if not stated
otherwise. We decompose the operator $\mathcal{A}(z)$ into a sum 
of the main term $\mathcal{A}_0 (z)$ and the remainder 
\begin{equation}\label{09.5}
 \mathcal{A} (z) = \mathcal{A}_0 (z) + \mathcal{R} (z) . 
\end{equation}
The  operators on the rhs of (\ref{09.5}) are self-adjoint and compact. The proof of
Theorem~\ref{th:1} is based on the following two theorems  
\begin{theorem}\label{th:09.2}
 For all $r_\varepsilon \in (0, 1/4)$ and $a >0 $ 
\begin{equation}
 \lim_{z \to 0} |\ln|\ln z^2 | |^{-1} \mathfrak{n}(\mathcal{A}_0 (z), a) = \frac
8{3\pi a} . 
\end{equation}
\end{theorem}
\begin{theorem}\label{th:09.3}
 For each $\varepsilon > 0$ one can choose $r_\varepsilon \in (0, 1/4)$ so that 
\begin{equation}
 \varlimsup_{z \to 0} |\ln|\ln z^2 | |^{-1} \mathfrak{n}_\mu (\mathcal{R} (z),
\varepsilon) < \varepsilon. 
\end{equation}
\end{theorem}
Sec.~\ref{sec:3} is devoted to the proof of Theorem~\ref{th:09.2}. The proof of Theorem~\ref{th:09.3}, which is
rather involved and uses the machinery of trace ideals \cite{17}, is given in Sec.~\ref{sec:4}.
Using Theorems~\ref{th:09.2}, \ref{th:09.3} we 
can prove the main theorem 
\begin{proof}[Proof of Theorem~\ref{th:1}]
For any given $\varepsilon \in (0, 1)$  let us fix $r_\varepsilon >0$ according to
Theorem~\ref{th:09.3}. The eigenvalue distribution function satisfies the inequality \cite{birman09,3} 
\begin{equation}\label{01.09}
 \mathfrak{n} (A_1 + A_2 , a_1 + a_2) \leq \mathfrak{n}  (A_1 , a_1 ) +
\mathfrak{n}  (A_2 ,  a_2) ,  
\end{equation}
where $A_{1,2} \in \mathcal{C}(\mathcal{H})$ and $a_{1,2} >0$. Using this inequality we obtain from (\ref{09.5}) 
\begin{gather}
 \mathfrak{n} (\mathcal{A} (z), 1) \leq \mathfrak{n}(\mathcal{A}_0 (z), 1-
\varepsilon) + \mathfrak{n}(\mathcal{R} (z), \varepsilon) \leq
\mathfrak{n}(\mathcal{A}_0 (z), 1- \varepsilon) + \mathfrak{n}_\mu (\mathcal{R}
(z), \varepsilon)\\
\mathfrak{n} (\mathcal{A}_0 (z), 1+\varepsilon) \leq \mathfrak{n}(\mathcal{A}
(z), 1) + \mathfrak{n}(-\mathcal{R} (z), \varepsilon) \leq
\mathfrak{n}(\mathcal{A} (z), 1) + \mathfrak{n}_\mu (\mathcal{R} (z),
\varepsilon). 
\end{gather}
(Because $\mathcal{R} (z)$ is self-adjoint we have $\mathfrak{n}(\pm \mathcal{R}
(z), \varepsilon) \leq \mathfrak{n}_\mu (\mathcal{R} (z), \varepsilon)$).
Hence, 
using Theorems~\ref{th:09.2}, \ref{th:09.3}  we get 
\begin{gather}
 \frac 8{3\pi(1+\varepsilon)} - \varepsilon \leq \varliminf_{z\to 0} |\ln|\ln
z^2 | |^{-1} \mathfrak{n}(\mathcal{A} (z), 1) \nonumber\\
\leq \varlimsup_{z\to 0} |\ln|\ln z^2 | |^{-1} \mathfrak{n}(\mathcal{A} (z), 1)
\leq \frac 8{3\pi(1-\varepsilon)} + \varepsilon . 
\end{gather}
Letting $\varepsilon \to 0$ and using (\ref{25}) we complete the proof. 
\end{proof}

\section{Spectral asymptotic of the leading term}\label{sec:3}

\begin{definition}\label{def:1}
For an operator function $B: \mathbb{R}_+ /\{0\} \to \mathcal{C}(\mathcal{H})$
we shall write $B(z) = \mathcal{O}_{C} (z)$ if and only if for all $\epsilon
>0$ 
there exists $z_0 >0$ and a decomposition $B(z) = B_\epsilon (z) + P_\epsilon
(z)$, where $B_\epsilon , P_\epsilon: \mathbb{R}_+ /\{0\} \to
\mathcal{C}(\mathcal{H})$ are such that 
$\sup_{z \in (0,z_0)} \|B_\epsilon (z)\| < \epsilon$ and $\sup_{z \in (0, z_0)}
\dim \Ran P_\epsilon (z) < \infty$. 
\end{definition}
\begin{remark}
 If $B(z) = \mathcal{O}_{C} (z)$ in Def.~\ref{def:1} is such that $B(z) = B^*
(z)$ then in the decomposition one can choose $B_\epsilon , P_\epsilon$ so that 
$B_\epsilon (z) = B_\epsilon^* (z)$ and $P_\epsilon (z) = P_\epsilon^* (z)$.
(This can be verified by writing $B(z) = (1/2)[B_\epsilon (z) + B^*_\epsilon
(z)] 
+ (1/2)[P_\epsilon (z) + P^*_\epsilon (z)]$). 
\end{remark}
The following proposition is obvious 
\begin{proposition}\label{prop:1}
 Suppose that $B: \mathbb{R}_+ /\{0\} \to \mathcal{C}(\mathcal{H})$ is such that
$\sup_{z \in (0, z_0)} \|B (z) \|_{HS} < \infty$ for some $z_0 >0$. 
Then $B(z) = \mathcal{O}_{C} (z)$. 
\end{proposition}
\begin{proof}
Let us write the singular value decomposition \cite{17}
\begin{equation}
 B(z) = \sum_{k=1}^\infty \mu_k \bigl(B(z)\bigr) \Bigl(\phi_k (z) , \cdot \Bigr)
\psi_k (z), 
\end{equation}
where $\{\phi_k (z)\}, \{\psi_k (z)\}$ are orthonormal sets. For $z \in (0,
z_0)$ the following inequality holds 
\begin{equation}
 n\mu_n^2 (z)\leq \mu_1^2 (z)+ \cdots + \mu_n^2 (z) \leq \|B (z) \|_{HS}^2 \leq
\alpha , 
\end{equation}
where $\alpha := \sup_{z \in (0, z_0)} \|B (z) \|_{HS}$. For any given $\epsilon
>0$ we can set $n $ equal to the integer, which is larger than $\alpha \epsilon^{-2}$. Then 
$B_\epsilon (z) := \sum_{k = n+1}^\infty \mu_k (z) \bigl(\phi_k (z) , \cdot
\bigr) \psi_k (z)$ and $P_\epsilon (z) := \sum_{k = 1}^n \mu_k (z) \bigl(\phi_k
(z) , \cdot \bigr) \psi_k (z)$ 
fulfill the requirement in Def.~\ref{def:1}.  
\end{proof}
\begin{definition}\label{def:2}
Consider two operator functions $A: \mathbb{R}_+ /\{0\} \to
\mathcal{C}(\mathcal{H}_1)$, $B: \mathbb{R}_+ /\{0\} \to
\mathcal{C}(\mathcal{H}_2)$ such that 
$A^*(z) = A(z)$, $B^*(z) = B(z)$.  We shall say that $A(z)$ and $B(z)$ are
equivalent and write $A(z) \sim B(z)$ if either $\lambda_n (A(z)) = \lambda_n
(B(z))$ for $n = 1, 2, \ldots$ or 
$A(z) - B(z) = \mathcal{O}_C(z)$ (the last case implies $\mathcal{H}_1 =
\mathcal{H}_2$). 
\end{definition}
Let us explain the point of Definition~\ref{def:2}. Below we shall prove that
$\mathcal{A}_0 (z) $ is equivalent to some operator function $T(z)$, whose 
spectrum is known explicitly. Then we shall prove that $\lim_{z
\to 0} |\ln |\ln z||^{-1} \mathfrak{n}(\mathcal{A}_0 (z) , a) = \lim_{z \to 0}
|\ln |\ln z||^{-1} \mathfrak{n}(T (z) , a)$ for $a >0$ 
thus proving Theorem~\ref{th:1}.

\begin{table}
\caption{\label{tab:1}Antisymmetry relations for $(\phi_1, \phi_2 , \phi_3 )\in
\mathcal{H}_A$.}
\begin{tabular}{|c|c|c|}
\hline
$p_{12}\phi_1 = -\phi_2$ & $p_{12}\phi_2 = -\phi_1$ & $p_{12}\phi_3 = -\phi_3$\\
\hline$p_{13}\phi_1 = -\phi_3$ & $p_{13}\phi_2 = -\phi_2$ & $p_{13}\phi_3 =
-\phi_1$\\
\hline$p_{13}\phi_1 = -\phi_1$ & $p_{13}\phi_2 = -\phi_3$ & $p_{23}\phi_3 =
-\phi_2$\\
\hline
\end{tabular}
\end{table}

Let us analyze the spectrum of $\mathcal{A}_0 (z)$. If $\mathcal{A}_0 (z) \Psi =
\lambda \Psi$, where 
$\Psi \in \mathcal{H}_A$ and $\lambda \neq 0$ then due to antisymmetry
requirements listed in Table~\ref{tab:1} we have 
\begin{equation}
 \Psi = 
\begin{pmatrix}
\mathcal{F}^{-1}_1 [ f_+ (q_1) \hat \eta_+ (p_1) + f_- (q_1) \hat \eta_- (p_1)]
\\ 
\mathcal{F}^{-1}_2 [ f_+ (q_2) \hat \eta_+ (p_2) + f_- (q_2) \hat \eta_- (p_2)]
\\ 
\mathcal{F}^{-1}_3 [ f_+ (q_3) \hat \eta_+ (p_3) + f_- (q_3) \hat \eta_- (p_3)]
 
\end{pmatrix}. 
\end{equation}
Substituting the last ansatz into the equation $\mathcal{A}_0 (z) \Psi = \lambda
\Psi$ and using (\ref{6}) we find that $f_\pm$ satisfy integral
equation 
\begin{equation}
L^{(a)}(z)
\begin{pmatrix}
f_+ \\
f_-
\end{pmatrix}
=
\begin{pmatrix}
 L^{(a)}_{11} (z) &L^{(a)}_{12} (z) \\ 
L^{(a)}_{21}(z) &L^{(a)}_{22} (z)
\end{pmatrix}
\begin{pmatrix}
f_+ \\
f_-
\end{pmatrix}
=\lambda 
\begin{pmatrix}
f_+ \\
f_-
\end{pmatrix} . 
\end{equation}
The matrix entries $L^{(a)}_{ik} (z)$ are integral operators on $L^2
(\mathbb{R}^2)$ with the kernels 
\begin{gather}
 L^{(a)}_{11} (p, q) =  \frac{-2 \hat \psi^*_+ \left(\frac 2{\sqrt 3} q + \frac
1{\sqrt 3} p \right)\hat \psi_+ \left(\frac 2{\sqrt 3} p + \frac 1{\sqrt 3} q
\right) g_z (|p|) g_z (|q|)}{(p^2 +q^2 +p\cdot q+z^2)} ,\label{54}\\
L^{(a)}_{12} (p, q) =  \frac{-2 \hat \psi^*_+ \left(\frac 2{\sqrt 3} q + \frac
1{\sqrt 3} p \right)\hat \psi_- \left(\frac 2{\sqrt 3} p + \frac 1{\sqrt 3} q
\right)g_z (|p|) g_z (|q|)}{(p^2 +q^2 +p\cdot q +z^2)} , \\
L^{(a)}_{22} (p, q) =  \frac{-2 \hat \psi^*_- \left(\frac 2{\sqrt 3} q + \frac
1{\sqrt 3} p \right)\hat \psi_- \left(\frac 2{\sqrt 3} p + \frac 1{\sqrt 3} q
\right)g_z (|p|) g_z (|q|)}{(p^2 +q^2 +p\cdot q+z^2)} , \label{56}
\end{gather}
and $L^{(a)}_{21} (p, q) = \bigl(L^{(a)}_{12} (q, p)\bigr)^*$. In
(\ref{54})-(\ref{56}) $\hat \psi_\pm (p)$ are Fourier transforms of the
functions 
\begin{equation}
 \psi_\pm (x) := |v(|x|)|^{\frac 12} \eta_\pm (x) . 
\end{equation}
The operator function $L^{(a)} (z)$ acts on $L^2 (\mathbb{R}^2) \oplus L^2
(\mathbb{R}^2) $ and clearly 
$L^{(a)} (z) \sim \mathcal{A}_0 (z)$. The relevant properties of the functions
$\hat \psi_\pm (p)$ are summarized in the following 
\begin{lemma}\label{lem:4}
In polar coordinates $ \hat \psi_\pm (p)= \psi_0 (|p|) e^{\pm i\varphi_p}$,
where $\psi_0 \in L^2 ((0, \infty); xdx)$. There are $\alpha, \beta, \gamma >0$
such that 
\begin{gather}
\left|\hat \psi_\pm (p) \right| \leq \alpha |p| , \label{58}\\
\left|\hat \psi_\pm (p+q) - \hat \psi_\pm (p) - \hat \psi_\pm (q) \right| \leq
\beta |p| |q| ,\label{59}\\
s^{-2}\psi^*_0 \left(\frac 1{\sqrt 3} s \right)  \psi_0 \left(\frac 2{\sqrt 3}s
\right) = \frac{c_0^2}6 + \omega (s) , \label{60}
\end{gather} 
where $|\omega (s)| \leq \gamma s^2 $ and $c^2_0$ is defined in (\ref{17}). 
\end{lemma}
\begin{proof}
 The representation in polar coordinates follows immediately from symmetry
arguments. The Fourier transformed function is 
\begin{equation}
 \hat \psi_\pm (p) := \frac 1{2\pi} \int e^{-ip\cdot r} \psi_\pm (r) d^2 r =
\frac 1{2\pi} \int (e^{-ip\cdot r} - 1)\psi_\pm (r) d^2 r . 
\end{equation}
Using that $|e^{-ip\cdot r} - 1| \leq 2$ we obtain 
\begin{equation}
 |\hat \psi_\pm (p)| \leq  |p|\pi^{-1}\left\||r||v(|r|)|^{\frac 12} \eta_\pm
(r)\right\|_1 \leq |p|\pi^{-1} \left\||r||v(|r|)|^{\frac 12} \right\|_2 . 
\end{equation}
(\ref{59}) is obtained similarly using the inequality 
\begin{equation}
 \left| e^{-i(p+q)\cdot r} - (e^{-ip\cdot r} - 1) - e^{-iq\cdot r} \right| = 
\left| ( e^{-iq\cdot r} - 1) (e^{-ip\cdot r} - 1)\right| \leq 4 |p| |q| r^2 . 
\end{equation}
Expanding the exponent in the Fourier transform we obtain the expression 
\begin{equation}
 \hat \psi_\pm (p)= \frac{-i}{2} |p| e^{\pm i\varphi_p}\int_0^\infty s^2 \eta_0
(s) |v (s)|^\frac 12 + \mathcal{O} (|p|^3) , 
\end{equation}
from which (\ref{60}) follows. 
\end{proof}

Using Lemmas~\ref{lem:4}, \ref{lem:6} we conclude that $L^{(a)}(z) \sim L^{(b)}
(z)$, 
where the matrix entries of $L^{(b)}(z)$ are integral operators with the
kernels 
\begin{gather}
 L^{(b)}_{11} (p, q) =  \frac{-2 \left[ \hat \psi^*_+ \left(\frac 1{\sqrt 3} p
\right)\hat \psi_+ \left(\frac 2{\sqrt 3} p \right) + \hat \psi^*_+ \left(\frac
2{\sqrt 3} q \right)\hat \psi_+ \left(\frac 1{\sqrt 3} q \right)\right]g_z (|p|)
g_z (|q|)}{ (p^2 +q^2 +p\cdot q+z^2)} ,\\
L^{(b)}_{12} (p, q) =  \frac{-2 \left[ \hat \psi^*_+ \left(\frac 1{\sqrt 3} p
\right)\hat \psi_- \left(\frac 2{\sqrt 3} p \right) + \hat \psi^*_+ \left(\frac
2{\sqrt 3} q \right)\hat \psi_- \left(\frac 1{\sqrt 3} q \right)\right]g_z (|p|)
g_z (|q|)}{ (p^2 +q^2 +p\cdot q+z^2)} ,\\
L^{(b)}_{22} (p, q) =  \frac{-2 \left[ \hat \psi^*_- \left(\frac 1{\sqrt 3} p
\right)\hat \psi_- \left(\frac 2{\sqrt 3} p \right) + \hat \psi^*_- \left(\frac
2{\sqrt 3} q \right)\hat \psi_- \left(\frac 1{\sqrt 3} q \right)\right]g_z (|p|)
g_z (|q|)}{(p^2 +q^2 +p\cdot q+z^2)} , 
\end{gather}
and $L^{(b)}_{21} (p, q) = \bigl(L^{(b)}_{12} (q, p)\bigr)^*$. The operator
$L^{(b)}(z)$ acts on  $L^2 (\mathbb{R}^2) \oplus L^2 (\mathbb{R}^2) $. 
By (\ref{58}) we have 
\begin{gather}
 \left|\hat \psi^*_+ \left((1/\sqrt 3) p \right)\hat \psi_+ \left((2/\sqrt 3) p
\right)\right| \left| \frac 1{p^2 +q^2 +p\cdot q+z^2} -  \frac 1{p^2 +q^2
}\right|  \nonumber\\
\leq \frac{2 \alpha^2}3  \frac {|p||q|}{p^2 +q^2 +p\cdot q} + \frac{2 \alpha^2}3
\frac {z^2}{p^2 +q^2 +p\cdot q + z^2} . 
\end{gather}
Hence, by Lemma~\ref{lem:6} we have $L^{(b)}(z) \sim L^{(c)} (z)$, 
where $L^{(c)}(z)$ has the matrix entries
\begin{gather}
 L^{(c)}_{11} (p, q) =  \frac{-2 \left[ \hat \psi^*_+ \left(\frac 1{\sqrt 3} p
\right)\hat \psi_+ \left(\frac 2{\sqrt 3} p \right) + \hat \psi^*_+ \left(\frac
2{\sqrt 3} q \right)\hat \psi_+ \left(\frac 1{\sqrt 3} q \right)\right]g_z (|p|)
g_z (|q|)}{(p^2 +q^2 )} ,\\
L^{(c)}_{12} (p, q) =  \frac{-2 \left[ \hat \psi^*_+ \left(\frac 1{\sqrt 3} p
\right)\hat \psi_- \left(\frac 2{\sqrt 3} p \right) + \hat \psi^*_+ \left(\frac
2{\sqrt 3} q \right)\hat \psi_- \left(\frac 1{\sqrt 3} q \right)\right]g_z (|p|)
g_z (|q|)}{  (p^2 +q^2) } ,\\
L^{(c)}_{22} (p, q) =  \frac{-2 \left[ \hat \psi^*_- \left(\frac 1{\sqrt 3} p
\right)\hat \psi_- \left(\frac 2{\sqrt 3} p \right) + \hat \psi^*_- \left(\frac
2{\sqrt 3} q \right)\hat \psi_- \left(\frac 1{\sqrt 3} q \right)\right]g_z (|p|)
g_z (|q|)}{ (p^2 +q^2 )} , 
\end{gather}
and $L^{(c)}_{21} (p, q) = \bigl(L^{(c)}_{12} (q, p)\bigr)^*$. 
We are interested in the nontrivial spectrum of the compact operator
$L^{(c)}(z)$ on $L^2 (\mathbb{R}^2) \oplus L^2 (\mathbb{R}^2)$, that is we look
for solutions 
of the equation 
\begin{equation}\label{73}
L^{(c)}(z)
\begin{pmatrix}
F_+ \\
F_-
\end{pmatrix}
=
\begin{pmatrix}
 L^{(c)}_{11} (z) &L^{(c)}_{12} (z) \\ 
L^{(c)}_{21}(z) &L^{(c)}_{22} (z)
\end{pmatrix}
\begin{pmatrix}
F_+ \\
F_-
\end{pmatrix}
=\lambda 
\begin{pmatrix}
F_+ \\
F_-
\end{pmatrix} ,  
\end{equation} 
where $\lambda \neq 0$. 
Now we employ the symmetry of integral equations and expand $F_+ , F_-$ as
follows
\begin{gather}
  F_+ (p) = \sum_{l = -\infty}^\infty e^{i(l-1)\varphi_p} f^{(l)}_+ (|p|) ,
\label{74}\\ 
 F_- (p) = \sum_{l = -\infty}^\infty  e^{i(l+1)\varphi_p} f^{(l)}_- (|p|) , 
\label{75}
\end{gather}
where $f^{(l)}_\pm (x) \in L^2 ([0,r_\varepsilon]; xdx)$. Substituting
(\ref{74}), (\ref{75}) into (\ref{73}) we find that $f^{(l)}_+ (|p|) = 0$ for
all $l$ except $l = \pm 1$. 
Thus we conclude that $L^{(c)} (z) \sim T^{+}(z) \oplus T^{-}(z)$, where 
$T^{\pm}(z)$ act on 
$L^2 ([0,r_\varepsilon]; xdx) \oplus L^2 ([0,r_\varepsilon]; xdx)$ and have the
structure 
\begin{equation}
T^{\pm}(z) = 
\begin{pmatrix}
 T^{\pm}_{11} (z) &T^{\pm}_{12} (z) \\ 
T^{\pm}_{21}(z) &T^{\pm}_{22} (z) 
\end{pmatrix}.  
\end{equation}
The matrix entries are integral operators on $L^2 ([0,r_\varepsilon]; xdx)$ with
the following kernels
\begin{gather}
 T^{+}_{11} (s, t) = T^{-}_{22} (s, t) =  (-4\pi)\frac{  \left[\psi^*_0
\left(\frac 1{\sqrt 3} s \right)  \psi_0 \left(\frac 2{\sqrt 3}s \right) +  
\psi^*_0 \left(\frac 2{\sqrt 3} t \right)  \psi_0 \left(\frac 1{\sqrt 3}
t\right)\right] g_z (s) g_z (t)}{(s^2 +t^2 )} ,\\
T^+_{12}(s, t) = T^-_{21}(s, t) = (-4\pi) \frac{  \psi_0 \left(\frac 1{\sqrt 3}
t \right)  \psi^*_0 \left(\frac 2{\sqrt 3}t \right) g_z(s) g_z (t)}{(s^2 +t^2 )}
,\\
T^+_{21}(s, t) = T^-_{12}(s, t) =  (-4\pi) \frac{  \psi^*_0 \left(\frac 1{\sqrt
3} s \right)  \psi_0 \left(\frac 2{\sqrt 3}s \right) g_z (s) g_z (t)}{(s^2 +t^2
)} , \\
  T^{+}_{22} (s, t) =  T^{-}_{11} (s, t) = 0.  
\end{gather}
Below we shall consider only the spectrum of $T^{+}(z)$. The spectrum of
$T^{-}(z)$ is considered analogously. 
Let us introduce two integral operator functions $B_{1,2} (z)$ on $L^2
([0,r_\varepsilon]; xdx)$ with the following integral kernels 
\begin{gather}
 B_1 (s,t) = \frac{  \psi^*_0 \left(\frac 1{\sqrt 3} s \right)  \psi_0
\left(\frac 2{\sqrt 3}s \right)g_z (s) g_z (t) \chi_{\{s <t\}}}{(s^2 +t^2 )} ,
\label{81}\\
B_2 (s,t) = \psi^*_0 \left(\frac 1{\sqrt 3} s \right)  \psi_0 \left(\frac
2{\sqrt 3}s \right)g_z (s) g_z (t) \chi_{\{s \geq t\}}\left[\frac{1}{s^2 +t^2} -
\frac{1}{s^2}\right]. 
\end{gather}
Our aim is to prove that $ B_{1,2} (z) = \mathcal{O}_C (z)$. The function
$\chi_{\{s <t\}}$ is such that $\chi_{\{s <t\}} = 1$ if $s<t$ and 
$\chi_{\{s <t\}} = 0$ if $s \geq t$ (the notation for 
this function using other relation symbols is self-explanatory). Using Lemma~\ref{lem:4} 
and (\ref{34}) for $z \in (0, r_\varepsilon]$ we get  
\begin{gather}
 \|B_1 (z)\|^2_{HS} \leq \delta^2 \alpha^4 \int_0^{r_\varepsilon}
\int_0^{r_\varepsilon}  \frac{s^5 t \chi_{\{s <t\}}ds dt}{(s^2 +z^2)|\ln (s^2
+z^2)| (t^2 +z^2)|\ln (t^2 +z^2)|(s^2 +t^2)^2} \nonumber\\
\leq  \delta^2 \alpha^4 \int_0^{r_\varepsilon} \int_0^{r_\varepsilon}  \frac{s^2
 ds dt}{|\ln s^2||\ln t^2 |(s^2 +t^2)^2} , 
\end{gather}
where we have used that for small $s, z \in (0, r_\varepsilon]$ one has $(s^2 +z^2)|\ln (s^2 +z^2)| \geq
s^2 |\ln s^2|$. In the last integral we pass to polar coordinates $s = \rho \sin
\phi$, 
$t = \rho \cos \phi$ and obtain the inequality 
\begin{equation}
  \|B_1 (z)\|^2_{HS} \leq   \frac{\delta^2 \alpha^4 }{4}\int_0^{2 r_\varepsilon}
\int_0^{\pi/2} \frac{\sin^2 \phi \; d\rho \: d\phi}{\rho|\ln \rho|^2} =
\frac{\delta^2 \alpha^4 \pi}{16} \int_0^{2 r_\varepsilon}  \frac{ d\rho
}{\rho|\ln \rho|^2}  . 
\end{equation}
The last integral converges \cite{15} and from Proposition~\ref{prop:1} it follows that $ B_1 (z) = \mathcal{O}_C (z)$. Similarly one shows that $ B_2 (z) = \mathcal{O}_C (z)$. 
Using this fact we obtain $T^+(z) \sim T^{(a)}
(z)$, where $T^{(a)} (z)$ acts on the same space as $T^+(z)$ and its matrix
entries 
have the following integral kernels 
\begin{gather}
 T^{(a)}_{11} (s, t) =  (-4\pi)\left[s^{-2}\psi^*_0 \left(\frac 1{\sqrt 3} s
\right)  \psi_0 \left(\frac 2{\sqrt 3}s \right)\chi_{\{s \geq t\}} \right.
\nonumber\\
\left. +   t^{-2}\psi^*_0 \left(\frac 2{\sqrt 3} t \right)  \psi_0 \left(\frac
1{\sqrt 3} t\right)\chi_{\{s \leq t\}}\right] g_z (s) g_z (t) , \label{87}\\
T^{(a)}_{12} (s, t) =  (-4\pi) t^{-2} \psi_0 \left(\frac 1{\sqrt 3} t \right) 
\psi^*_0 \left(\frac 2{\sqrt 3}t \right)  \chi_{\{s \leq t\}} g_z (s) g_z (t)  ,
\\
T^{(a)}_{21} (s, t) = (-4\pi) s^{-2}\psi^*_0 \left(\frac 1{\sqrt 3} s \right) 
\psi_0 \left(\frac 2{\sqrt 3}s \right) \chi_{\{s \geq t\}}g_z (s) g_z (t) 
\label{89}, 
\end{gather}
$ T^{(a)}_{22} (s, t) = 0$. 
Now let us consider the expression in square brackets in (\ref{87}). Due to
(\ref{60}) we have a. e.
\begin{gather}
 s^{-2}\psi^*_0 \left(\frac 1{\sqrt 3}  s \right)  \psi_0 \left(\frac 2{\sqrt
3}s \right)\chi_{\{s \geq t\}} +   t^{-2}\psi^*_0 \left(\frac 2{\sqrt 3} t
\right)  \psi_0 \left(\frac 1{\sqrt 3} t\right)\chi_{\{s \leq t\}} \nonumber\\
= \frac {c_0^2}6 + \omega(s) \chi_{\{s \geq t\}} + \chi_{\{s \leq t\}} \omega(t)
= \frac {c_0^2}6 + \omega(s) + \omega(t) - \left[\omega(s) \chi_{\{s \leq t\}} +
\chi_{\{s \geq t\}} \omega(t)\right] . 
\end{gather}
Using (\ref{34}) and (\ref{60}) one  can easily check that $ B_3 (z) 
=\mathcal{O}_C (z)$, where the integral operator 
$B_3 (z)$ acts on $L^2 ([0,r_\varepsilon]; xdx)$ and has the kernel 
\begin{equation}
 B_3 (s,t) = \omega(s)\chi_{\{s \leq t\}}g_z (s) g_z (t) . 
\end{equation}
Making similar decompositions for other kernels in (\ref{87})-(\ref{89}) and
using that $ B_3 (z) =\mathcal{O}_C (z)$ we conclude that 
\begin{equation}\label{93}
 T^{(a)} (z) \sim S(z) - \frac{2\pi c_0^2}{3}\mathcal{T}(z) , 
\end{equation}
where $S(z), \mathcal{T}(z)$ act on the same space as $T^{(a)}(z)$ and their
matrix entries have the following integral kernels 
\begin{gather}
 S_{11}(s,t) = (-4\pi) \left[\frac{c_0^2}6 + \omega(s) + \omega(t)\right]g_z (s)
g_z (t) \\
 S_{12}(s,t) = (-4\pi)\omega(t) g_z (s) g_z (t) ,\\
 S_{21}(s,t) = (-4\pi)\omega(s) g_z (s) g_z (t) ,
\end{gather}
and 
\begin{gather}
\mathcal{T}_{12} (s, t) =   \chi_{\{s \leq t\}} g_z(s) g_z(t), \\
\mathcal{T}_{21} (s, t) =  \chi_{\{s \geq t\}} g_z(s) g_z(t),\\
\mathcal{T}_{11} (s, t) = \mathcal{T}_{22} (s, t) =  S_{22}(s,t) = 0. 
\end{gather}
Finally, from (\ref{93}) we conclude that 
\begin{equation}\label{100}
 T^+ (z) \sim - \frac{2\pi c_0^2}3 \mathcal{T}(z)
\end{equation}
because $S(z)$ for all $z >0$ is a rank 3 operator. 
The nonzero spectrum of the operator $\mathcal{T}(z)$ can be calculated
explicitly. Note that $\sigma (\mathcal{T}(z))/\{0\} = \sigma
(\mathcal{T'}(z))/\{0\}$, where 
\begin{gather}
\mathcal{T'}_{12} (s, t) =  \chi_{\{s \leq t\}} g_z^2 (t) ,\\
\mathcal{T'}_{21} (s, t) =  \chi_{\{s \geq t\}} g_z^2 (t),\\
\mathcal{T'}_{11} (s, t) = \mathcal{T'}_{22} (s, t) = 0. 
\end{gather}
(This is the consequence of the fact that $\sigma (AB)/\{0\} = \sigma
(BA)/\{0\}$ for any bounded $A, B$, see \cite{16}). The equation 
$\mathcal{T'}(z) f = \lambda f$ for $\lambda \neq 0$ takes the form 
\begin{gather}
 \lambda f_1 (s) = \int_s^{r_\varepsilon} f_2 (t) g_z^2 (t) tdt \label{104}\\
 \lambda f_2 (s) = \int_0^s f_1 (t) g_z^2 (t) tdt. \label{105}
\end{gather}
Let us make the change of variables in (\ref{104})-(\ref{105}) setting $x = \xi
(s)$, where 
\begin{equation}
 \xi(s) := \int_0^s g_z^2(t) t dt \label{106}
\end{equation}
is monotone increasing. Then Eqs. (\ref{104})-(\ref{105}) take the form 
\begin{gather}
 \lambda \tilde f_1 (x) = \int_x^{\xi(r_\varepsilon)} \tilde f_2 (x') dx'
\label{107}\\
 \lambda \tilde f_2 (x) = \int_0^x \tilde f_1 (x') dx' ,  \label{108}
\end{gather}
where $\tilde f_i \in C^1 ([0, \xi(r_\varepsilon)])$. 
Similar integral equations were obtained in \cite{4,jesper}. Differentiating
(\ref{107})-(\ref{108}) over $x$ gives $\lambda (d\tilde f_1/dx) = - \tilde f_2
(x)$; $\lambda (d\tilde f_2/dx) = \tilde f_1 (x)$. 
These differential equations are solved by $\tilde f_1 (x) = \cos (\lambda
^{-1}x + \varphi_\lambda)$, and $\tilde f_2 (x) = \sin (\lambda ^{-1}x +
\varphi_\lambda)$. 
Substituting these expressions back into (\ref{107})-(\ref{108}) we find that
$\mathcal{T}(z)$ has an infinite number 
of positive and negative eigenvalues, namely, 
\begin{equation}\label{109}
 \lambda_k \bigl(\mathcal{T}(z)\bigr)= \lambda_k \bigl(-\mathcal{T}(z)\bigr) =
\frac{\xi(r_\varepsilon)}{(\pi/2) + \pi (k-1)} , \quad \textnormal{ where $k=1,
2, \ldots$}. 
\end{equation}
Note that 
\begin{equation}\label{volos13}
 \lim_{z \to 0} \left| \ln |\ln z^2 |\right|^{-1} \xi(r_\varepsilon )= \frac
2{\pi c_0^2} , 
\end{equation}
where $c_0^2$ is defined in (\ref{17}). Indeed, for any $\rho \in (0, r_\varepsilon)$ 
\begin{equation}
  \lim_{z \to 0} \left| \ln |\ln z^2 |\right|^{-1} \bigl[ \xi(r_\varepsilon ) - \xi(\rho )  \bigr] = 0. 
\end{equation}
Due to (\ref{34}) and (\ref{106}) we obtain 
\begin{equation}
 \frac{\delta ' }2  \leq \lim_{z \to 0} \left| \ln |\ln z^2 |\right|^{-1} \xi(\rho ) \leq  \frac{\delta}2 , 
\end{equation}
where we have used that 
\begin{equation}
 \lim_{z \to 0} \left| \ln |\ln z^2 |\right|^{-1} \int_0^\rho  \frac{t dt}{(t^2 +z^2)|\ln (t^2 +z^2)|} = \frac 12 . 
\end{equation}
Recall that for $\rho \to 0$ we have $\delta = 4/(\pi c_0^2) + \hbox{o}(\rho)$ and $\delta' = 4/(\pi c_0^2) +
\hbox{o}(\rho)$ (see the text below Eq.~(\ref{34})), which results in (\ref{volos13}). 
Now let us prove 
\begin{lemma}\label{lem:5}
 Suppoze $K:\mathbb{R}_+/\{0\} \to \mathcal{C}(\mathcal{H})$ is such that $K(z)
\sim - \mathcal{T}(z)$. Then for any $a >0$ 
\begin{equation}
\lim_{z \to 0} \left| \ln |\ln z^2 |\right|^{-1}\mathfrak{n}(K(z), a ) = \frac 2{\pi^2 c_0^2 a}. 
\end{equation}
\end{lemma}
\begin{proof}
From (\ref{volos13}) it follows that  $\xi(r_\varepsilon ) \to + \infty$
when $z \to 0$. Hence, from (\ref{109}) we get 
\begin{equation}\label{09.11}
  \lim_{z \to 0} [\xi(r_\varepsilon )]^{-1}\mathfrak{n}(-\mathcal{T}(z), a ) = \frac 1{\pi a} . 
\end{equation}
We only have to consider the case when $\mathcal{H} = L^2 ([0,r_\varepsilon];
xdx) \oplus L^2 ([0,r_\varepsilon]; xdx)$, otherwise the statement is obvious. 
For any fixed $\epsilon \in (0, a/2)$  
there exist $z_0, k >0$, and self-adjoint $B_\epsilon, P_\epsilon : \mathbb{R}_+
\to \mathcal{C}(\mathcal{H})$ such that $K(z) =
-\mathcal{T}(z) + B_\epsilon (z) + P_\epsilon (z)$, whereby $\sup_{z \in z_0} \|B_\epsilon (z)\| <
\epsilon$, and 
$\sup_{z \in (0, z_0)} \dim \Ran P_\epsilon (z) < k$. By (\ref{01.09})
\begin{gather}
 \mathfrak{n} (-\mathcal{T}(z), a + 2\epsilon) \leq \mathfrak{n} (K(z), a ) + \mathfrak{n} (-B_\epsilon (z), \epsilon ) + \mathfrak{n} (-P_\epsilon (z), \epsilon ) \label{09.12}\\
 \mathfrak{n} (K(z), a ) \leq \mathfrak{n} (-\mathcal{T}(z), a - 2\epsilon)  + \mathfrak{n} (B_\epsilon (z), \epsilon ) + \mathfrak{n} (P_\epsilon (z), \epsilon ) .\label{09.13} 
\end{gather}
Since $\mathfrak{n} (\pm B_\epsilon (z), \epsilon ) = 0$ and $\mathfrak{n} (P_\epsilon (z), \epsilon ) \leq k$ we obtain from (\ref{09.12}), (\ref{09.13}) and (\ref{09.11})
\begin{equation}
 \frac 1{\pi (a+2\epsilon)} \leq \varliminf_{z \to 0} [\xi(r_\varepsilon )]^{-1}\mathfrak{n}(K(z), a ) \leq  \varlimsup_{z \to 0} [\xi(r_\varepsilon )]^{-1}\mathfrak{n}(K(z), a ) \leq \frac 1{\pi (a-2\epsilon)} . 
\end{equation}
Letting $\epsilon \to 0$ we prove that 
\begin{equation}
 \lim_{z \to 0} [\xi(r_\varepsilon )]^{-1}\mathfrak{n}(K(z), a ) = \frac 1{\pi a} . 
\end{equation}
Now the result follows from (\ref{volos13}). 
\end{proof}
\begin{proof}[Proof of Theorem~\ref{th:09.2}]
 From Lemma~\ref{lem:5} and (\ref{100}) it follows that 
\begin{equation}\label{124}
\lim_{z \to 0} \left| \ln |\ln z^2 |\right|^{-1} \mathfrak{n}(T^\pm (z), a ) = 
\lim_{z \to 0} \left| \ln |\ln z^2 |\right|^{-1} \mathfrak{n}\left (\frac 3 {2\pi c_0^2} T^\pm (z), \frac {3a}{2\pi c_0^2} \right) = \frac 4{3\pi a}
\end{equation}
(we have proved (\ref{124}) for $T^+ (z)$, but the analysis of the operator $T^-
(z)$ leads to the same result). By the above analysis 
$\mathcal{A}_0 (z)\sim T^+ (z) \oplus T^- (z)$. Repeating the arguments in the
proof of Lemma~\ref{lem:5} we obtain 
\begin{equation}\label{volos12}
\lim_{z \to 0} \left| \ln |\ln z^2 |\right|^{-1} \mathfrak{n}(\mathcal{A}_0 (z), a ) =
\lim_{z \to 0} \left| \ln |\ln z^2 |\right|^{-1}  \Bigl( \mathfrak{n}(T^+ (z), a ) + 
\mathfrak{n}(T^- (z), a )\Bigr) = \frac 8{3\pi a} . 
\end{equation}
\end{proof}
The proof of the following lemma uses the idea in \cite{2}. 
\begin{lemma}\label{lem:6}
Suppose that the integral operator functions $C_{1,2} (z) : \mathbb{R}_+/\{0\}
\to \mathcal{C}(L^2(\mathbb{R}^2))$ have the integral kernels 
\begin{gather}
 C_1(p,p') = \frac{z^2g_z(|p|)g_z(|p'|)}{p^2 + p'^2 + p\cdot p' + z^2} , \\
C_2(p,p') = \frac{|p||p'|g_z(|p|)g_z(|p'|)}{p^2 + p'^2 + p\cdot p' + z^2} . 
\end{gather}
Then $C_{1,2}(z)  = \mathcal{O}_C (z)$. 
\end{lemma}
\begin{proof}
By a direct check one finds that $\sup_{z>0} \|C_1(z)\|_{HS} < \infty$, hence $C_1 (z)  = \mathcal{O}_C (z)$ by Proposition~\ref{prop:1}. 
Consider the integral operator $\tilde C$ on $L^2(\mathbb{R}^2)$ with the kernel 
\begin{equation}
 \tilde C (p,p') = \frac{\chi_{[0,1]}(|p|)\chi_{[0,1]}(|p'|)}{p^2 + p'^2 } . 
\end{equation}
Let us show that this operator is bounded. Using the expansion like in (\ref{74})-(\ref{75}) we reduce the
problem to proving that the integral operator $D$ on $L^2 ((0,1);xdx)$ with the
kernel 
\begin{equation}
 D(x,x') = \frac 1{x^2 + x'^2 } 
\end{equation}
is bounded. Consider the operator $W: L^2 ((0,1);xdx) \to L^2 (0,\infty)$, which
acts on  $f \in L^2 ((0,1);xdx)$ according to the rule 
$[Wf] (t) = e^{-t}f\bigl(e^{-t}\bigr)$. The operator $W$ has a
bounded inverse and $\|Wf\|=\|f\|$, which means that $W$ is unitary. 
The operator $WDW^{-1} : L^2 (0,\infty) \to L^2 (0,\infty)$ acts on $f(t)$ in
the following way 
\begin{equation}
 [WDW^{-1} f] (t) = \frac 12 \int_0^\infty \frac{f(t')dt'}{\cosh (t-t')} . 
\end{equation}
Applying the Young inequality \cite{11} we get 
\begin{equation}
 \|D\| = \|WDW^{-1}\| \leq \frac 12 \int_0^\infty \frac{dx}{\cosh x} =
\frac{\pi}4 . 
\end{equation}
Now let us write $C_2 (z)$ in the form 
\begin{gather}
 C_2 (z) = \chi_{(r, \infty)} (|p|) C_2 (z) \chi_{(r, \infty)} (|p|)+ \bigl\{ \chi_{[0, r]}(|p|) C_2 (z) \chi_{(r, \infty)}  (|p|)\nonumber \\
+ \chi_{(r, \infty)} (|p|)C_2 (z) \chi_{[0, r]} (|p|) + \chi_{[0, r]} (|p|) C_2 (z) \chi_{[0, r]} (|p|)\bigr\} . 
\end{gather}
On one hand, $\chi_{(r, \infty)} (|p|) C_2 (z) \chi_{(r, \infty)} (|p|) = \mathcal{O}_C (z)$ by Proposition~\ref{prop:1}. On the other hand, the norm of the 
terms in curly brackets can be made as small as pleased by choosing $r$ small enough (this easily follows from (\ref{34}) and the fact that $\tilde C$ is bounded). 
Hence, $ C_2 (z)  = \mathcal{O}_C (z)$. 
\end{proof}

\section{Spectral bounds for the remainder}\label{sec:4}

Suppose that $A_{1,2} \in \mathcal{C}(\mathcal{H})$ and $a_{1,2} >0$. Then the
distribution function of singular values satisfies the inequality 
\begin{equation}\label{25.08}
 \mathfrak{n}_\mu (A_1 + A_2 , a_1 + a_2) \leq \mathfrak{n}_\mu (A_1 , a_1 ) +
\mathfrak{n}_\mu (A_2 ,  a_2) . 
\end{equation}
The proof of (\ref{25.08}) can be found in \cite{birman09} (see page 245). Using inequalities
(1.4a), (1.4b) in \cite{17} one can easily show that 
\begin{gather}
\mathfrak{n}_\mu (A B, a) \leq  \mathfrak{n}_\mu (A , a\|B\|^{-1}) , \nonumber\\
\mathfrak{n}_\mu (BA, a) \leq  \mathfrak{n}_\mu (A , a\|B\|^{-1})  \label{09.94}
\end{gather}
for any bounded $B$ and $A \in \mathcal{C}(\mathcal{H})$. Following \cite{17} we shall
denote by $J_p$ normed trace ideals of compact operators, recall that the norm
of 
$A \in J_p$ is $\|A\|_p = \bigl( \sum_n \mu_n^p (A)\bigr)^{1/p}$. The trace
ideal $J_2$ is the family of Hilbert-Schmidt operators and $\|\cdot\|_2
\equiv\|\cdot\|_{HS}$. 
For $A \in J_p $, where $p \in [1,\infty)$ 
\begin{equation}\label{traceid}
 \mathfrak{n}_\mu (A , a ) = \mathfrak{n}_\mu (A^* , a ) \leq a^{-p} \|A\|_p^p . 
\end{equation}
Indeed, 
\begin{equation}
 \mathfrak{n}_\mu (A , a ) = \mathfrak{n}_\mu (a^{-1}A , 1 ) \leq \sum_n \mu_n^p
(a^{-1} A) = a^{-p} \|A\|_p^p. 
\end{equation}

Let us introduce the projection operator on $\mathcal{H}_A$ 
\begin{equation}
 \mathbb{P}_\pm (z) = \diag \left\{\mathcal{F}_1^{-1} \mathbb{P}_\pm^{(1)}
(z)\mathcal{F}_1, \mathcal{F}_2^{-1} \mathbb{P}_\pm^{(2)} (z) \mathcal{F}_2,
\mathcal{F}_3^{-1} \mathbb{P}_\pm^{(3)} (z)\mathcal{F}_3\right\}, 
\end{equation}
where $\mathbb{P}_\pm^{(k)} (z)$ act on $f(p_k, q_k)$ as follows 
\begin{equation}
  [\mathbb{P}_\pm^{(k)} (z) f ] (p_k , q_k) = \hat \varphi_\pm (\sqrt{z^2 +
q_k^2}; p_k) \int \hat \varphi_\pm^* (\sqrt{z^2 + q_k^2}; p'_k) f (p'_k , q_k)
dp'_k , 
\end{equation}
and $\hat \varphi_\pm (z; p)$ is the Fourier transform of $\varphi_\pm (z)$ in
(\ref{13}). Let us denote $\mathbb{P} (z) = \mathbb{P}_+ (z) + \mathbb{P}_- (z)$
and 
$\mathbb{Q} (z) = 1 - \mathbb{P} (z)$, and similarly $P = P_+ + P_-$, where $P_\pm $ were defined in (\ref{09:15}). 
Using the cutoff operator in (\ref{cutoff25}) we can
write 
the decomposition 
\begin{equation}
 \mathcal{A} (z) = \mathcal{\tilde A} (z) + \mathcal{R}_1 (z) + \mathcal{R}_1^*
(z) + \mathcal{R}_2 (z) + \mathcal{R}_3 (z) + \mathcal{R}_3^* (z) +
\mathcal{R}_4 (z) , 
\end{equation}
where 
\begin{gather}
 \mathcal{\tilde A} (z) = \mathfrak{X}_{[0, r_\varepsilon]} \mathbb{P} (z)
\mathcal{A} (z) \mathbb{P} (z) \mathfrak{X}_{[0, r_\varepsilon]} \nonumber\\
\mathcal{R}_1 (z) = \mathfrak{X}_{[0, r_\varepsilon]} \mathbb{Q} (z) \mathcal{A}
(z) \mathbb{P} (z) \mathfrak{X}_{[0, r_\varepsilon]} \nonumber\\
\mathcal{R}_2 (z) = \mathfrak{X}_{[0, r_\varepsilon]} \mathbb{Q} (z) \mathcal{A}
(z) \mathbb{Q} (z) \mathfrak{X}_{[0, r_\varepsilon]} \nonumber\\
\mathcal{R}_3 (z) =  \mathfrak{X}_{[0, r_\varepsilon]}\mathcal{A} (z) 
\mathfrak{X}_{(r_\varepsilon, \infty)} \nonumber\\
\mathcal{R}_4 (z) =  \mathfrak{X}_{(r_\varepsilon, \infty)} \mathcal{A} (z) 
\mathfrak{X}_{(r_\varepsilon, \infty)} \nonumber
\end{gather}
The decomposition (\ref{09.5}) holds true if we set 
\begin{equation}\label{09.22}
 \mathcal{R} (z) = \sum_{k=0}^2 \mathcal{R}_{2k} (z) + \sum_{k=1}^2
\bigl[\mathcal{R}_{2k-1} (z)+ \mathcal{R}_{2k-1}^* (z)\bigr] , 
\end{equation}
where by definition 
\begin{equation}\label{09.41}
 \mathcal{R}_{0} (z) = \mathcal{\tilde A} (z) - \mathcal{A}_0 (z) . 
\end{equation}
The proof of Theorem~\ref{th:09.3} is based on the following two lemmas 
\begin{lemma}\label{newlem:5}
For all $\varepsilon >0$ one can always fix $r_\varepsilon \in (0, 1/4)$ so that 
\begin{equation}
  \varlimsup_{z \to 0} |\ln|\ln z^2 | |^{-1} \mathfrak{n}_\mu (\mathcal{R}_i
(z), \varepsilon) < \varepsilon \quad \quad (i=0,1,2). \label{09.71}
\end{equation}
\end{lemma}
\begin{lemma}\label{newlem:6}
For any fixed $r_\varepsilon \in (0, 1/4)$ and $\varepsilon >0$  
\begin{equation}
  \varlimsup_{z \to 0} |\ln|\ln z^2 | |^{-1} \mathfrak{n}_\mu (\mathcal{R}_i
(z), \varepsilon) =0 \quad \quad (i=3,4). \label{09.72}
\end{equation}
\end{lemma}
\begin{proof}[Proof of Theorem~\ref{th:09.3}]
 By (\ref{09.22}) and (\ref{25.08}) 
\begin{equation}
 \mathfrak{n}_\mu (\mathcal{R} (z), \varepsilon) \leq \sum_{k=0}^2 
\mathfrak{n}_\mu (\mathcal{R}_{2k} (z), \varepsilon/7) + \sum_{k=1}^2 2
\mathfrak{n}_\mu (\mathcal{R}_{2k-1} (z), \varepsilon/7) . 
\end{equation}
Let us fix $r_\varepsilon$ as in Lemma~\ref{newlem:5}. Then by Lemmas~\ref{newlem:5}, \ref{newlem:6} 
\begin{equation}
\varlimsup_{z \to 0} |\ln|\ln z^2 | |^{-1}  \mathfrak{n}_\mu (\mathcal{R} (z),
\varepsilon)< 4 \varepsilon/7 . 
\end{equation}
\end{proof}

We shall need the following estimates of the Hilbert-Schmidt operator norms  
\begin{lemma}\label{lem:2}
For $z \in (0, r_\varepsilon ]$ and $R>1$ there is $c > 0$ such that 
\begin{gather}
 \left\|\mathfrak{X}_{[0,r_\varepsilon]} M' (z) G (z) \right\|_{HS}^2 \leq c r_\varepsilon^2 |\ln|\ln z^2 || , \label{09.02}\\
 \left\|\mathfrak{X}_{[0,r_\varepsilon]} \mathbb{P}(z) \mathcal{A}(z) \mathfrak{X}_{(R, \infty)}  \right\|_{HS}^2 \leq c R^{-2} |\ln|\ln z^2 || . \label{09.03}
\end{gather}
\end{lemma}
Before proving Lemma~\ref{lem:2} let us prove the following trivial bound 
\begin{lemma}\label{lem:3}
 Suppose that $v_0: \mathbb{R}^2 \to \mathbb{R}_+$ is Borel and $|v_0 (x)| \leq
\alpha_1 e^{-\alpha_2 |x|}$, where $\alpha_{1,2}$ are constants. 
Then its Fourier transform $\hat v_0  (p)$ for any 
$a \in \mathbb{R}^2$  satisfies the inequality 
\begin{equation}
 \left| \hat v_0 (p + a) - \hat v_0 (p-a)\right| \leq \left[\min (1, |a|)\right]
|\hat f_a (p)| , 
\end{equation}
where $\hat f_a (p)$ is the Fourier-transform of $f_a \in L^2 (\mathbb{R}^2)$
and $\sup_a \|f_a\| < \infty$. 
\end{lemma}
\begin{proof}
By definition of the Fourier transform
\begin{equation}
 \hat v_0 (p + a) - \hat v_0 (p-a)= \frac 1{\pi} \int e^{-i(p \cdot x)}
\left[e^{-i(a \cdot x)} - e^{i(a \cdot x)} \right] v_0 (x) d^2 x = [\min (1,
|a|)] \hat f_a ( p) , 
\end{equation}
where 
\begin{equation}\label{38}
f_a (x) := -2[\min (1, |a|)]^{-1}\sin (a\cdot x) e^{- \frac{\alpha_2}2  |x|}
\left[e^{\frac{\alpha_2}2  |x|} v_0 (x)\right] . 
\end{equation}
Note that 
\begin{equation}
 \sup_{a \neq 0} \left\|[\min (1, |a|)]^{-1}\sin (a\cdot x)  e^{-
\frac{\alpha_2}2  |x|}\right\|_\infty \leq \left\|(1+|x|)e^{- \frac{\alpha_2}2 
|x|}\right\|_\infty < \infty. 
\end{equation}
The Lemma is proved because the norm of the function in square brackets in
(\ref{38}) is finite. 
\end{proof}
\begin{proof}[Proof of Lemma~\ref{lem:2}]
 Let us start with (\ref{09.02}).  Without loss of generality and in view of antisymmetry relations
(Table~\ref{tab:1}) it is enough to prove that 
\begin{equation}\label{40}
 \left\|(1-p_{23}) \mathcal{F}_1^{-1}\chi_{[0,r_\varepsilon]} (|q_1 |)\mathcal{F}_1 M_{12} (z) (1-p_{13})
\mathcal{F}_2^{-1}g_z (|q_2 |)\mathcal{F}_2 \right\|^2_{HS} \leq c r_\varepsilon^2 
|\ln|\ln z^2 || , 
\end{equation}
for some $c >0$, where the operator in (\ref{40}) is considered on $L^2
(\mathbb{R}^4)$. 
After applying the appropriate Fourier transform the operator in (\ref{40}) has
the following integral kernel (c.f. eq. (\ref{22})) 
\begin{gather}
 K (p, q; p', q' ) = \frac 1{\pi^2} \chi_{[0,r_\varepsilon]} (|q |)  \frac {g_z (|q'|)}{(2q' + q)^2 + 3q^2  + 3
z^2}\nonumber\\
\times \left[ \widehat{|v|^{\frac 12 }} \Bigl(p + \frac{2}{\sqrt 3} q' +
\frac{1}{\sqrt 3} q\Bigr) - \widehat{|v|^{\frac 12 }} \Bigl(p - \frac{2}{\sqrt
3} q' - \frac{1}{\sqrt 3} q\Bigr)\right] \nonumber\\
\times \left[\widehat{|v|^{\frac 12 }} \Bigl(p' - \frac{1}{\sqrt 3} q' -
\frac{2}{\sqrt 3}  q \Bigr) - \widehat{|v|^{\frac 12 }} \Bigl(p' +
\frac{1}{\sqrt 3} q' + \frac{2}{\sqrt 3}  q \Bigr)\right].
\end{gather}
Thus by Lemma~\ref{lem:3}
\begin{gather}
 \int \left| K (p, q; p', q' )\right|^2 d^2 p d^2 q d^2 p' d^2 q' \nonumber\\
\leq c' \int \chi_{[0,r_\varepsilon]} (|q |)  g^2_z (|q'|) \frac{|2q'+q|^2|q'+2q|^2}{[(2q' + q)^2 + 3q^2 
+ 3 z^2]^2} d^2 q d^2 q' \leq c' (\alpha')^2 r_\varepsilon^2 \int g^2_z (|p|) d^2p , \label{09.32}
\end{gather}
where $c', \alpha'$ are constants and 
\begin{equation}
 \alpha' = \sup_{q', q \in \mathbb{R}^2} \frac{|2q'+q||q'+2q|}{(2q' + q)^2 + 3q^2  +
3 z^2} < \infty. 
\end{equation}
On account of (\ref{34}) for $z \leq r_\varepsilon$
\begin{equation}\label{09.31}
 \int g^2_z (|p|) d^2p \leq (2\pi)\delta \int_0^{r_\varepsilon} \frac{tdt}{(t^2
+z^2)|\ln (t^2 + z^2)|} = \pi  \delta \left[ \bigl|\ln|\ln z^2|\bigr| - \bigl|\ln|\ln (z^2 + r_\varepsilon^2)|\bigr| \right]
\end{equation}
Substituting (\ref{09.31}) into (\ref{09.32}) we prove (\ref{09.02}). Now let us consider (\ref{09.03}). 
\begin{gather}
  \left\|\mathfrak{X}_{[0,r_\varepsilon]} \mathbb{P}(z) \mathcal{A}(z) \mathfrak{X}_{(R, \infty)}  \right\|_{HS}^2 = 
\left\|\mathbb{P}(z)  G(z) M' (z) \mathfrak{X}_{(R, \infty)}  \bigl(1- M_d (z)\bigr)^{-1/2}\right\|_{HS}^2 \nonumber \\
\leq c' \left\|G(z) M' (z) \mathfrak{X}_{(R, \infty)}  \right\|_{HS}^2 , 
\end{gather}
where 
\begin{equation}
 c' = \sup_{z \in (0, r_\varepsilon]} \left\| \bigl(1- M_d (z)\bigr)^{-1/2} \mathfrak{X}_{(R, \infty)} \right\|^2 < \infty . 
\end{equation}
Hence, (\ref{09.03}) would follow if we could prove that 
\begin{equation}
  \left\|\mathfrak{X}_{(R, \infty)} M' (z)  G(z) \right\|_{HS}^2 \leq c R^{-2} |\ln|\ln z^2 ||
\end{equation}
for some $c >0$ (in the last equation we have used that $\|A\|_{HS}= \|A^*\|_{HS}$ for $A \in \mathcal{C}(\mathcal{H})$). Again without loss of generality 
the problem reduces to proving that 
\begin{equation}
 \left\|(1-p_{23}) \mathcal{F}_1^{-1}\chi_{(R, \infty)} (|q_1 |)\mathcal{F}_1 M_{12} (z) (1-p_{13})
\mathcal{F}_2^{-1}g_z (|q_2 |)\mathcal{F}_2 \right\|^2_{HS} \leq c R^{-2}
|\ln|\ln z^2 || , \label{09.37}
\end{equation}
for some $c>0$. Let us denote by $\tilde K(p,q;p', q')$ the integral kernel of the operator on the lhs of (\ref{09.37}). By Lemma 8 
we obtain 
\begin{gather}
 \int |\tilde K(p,q;p', q')| d^2 p d^2 q d^2 p' d^2 q' \nonumber\\\leq
\tilde c \int_{|q| > R} \frac{d^2 q}{[q^2 + z^2]^2} \int g_z^2 (|q'|) d^2 q' \leq c R^{-2} |\ln|\ln z^2 ||, 
\end{gather}
where $c, \tilde c >0$ are constants. 
\end{proof}
\begin{proof}[Proof of Lemma~\ref{newlem:5}]
Let us first consider $\mathcal{R}_0 (z)$, which according to (\ref{09.41}) can be rewritten
as 
\begin{equation}
 \mathcal{R}_0 (z) = \mathbb{P}(z) G(z) M' (z) G(z) \mathbb{P}(z) - PG(z) M'(z) G(z) P. 
\end{equation}
By (\ref{25.08}) without loosing generality it suffices to prove that for
any $\varepsilon >0$ one can choose $r_\varepsilon$ so that 
\begin{equation}
   \varlimsup_{z \to 0} |\ln|\ln z^2 | |^{-1} \mathfrak{n}_\mu
\Bigl(\bigl[\mathbb{P}(z) -P\bigr]G(z) M' (z) G(z) \mathbb{P}(z) ,
\varepsilon\Bigr) < \varepsilon/3. 
\end{equation}
Using the upper bound (\ref{traceid}) we get 
\begin{gather}
 \mathfrak{n}_\mu \Bigl(\bigl[\mathbb{P}(z) -P\bigr]G(z) M' (z) G(z)
\mathbb{P}(z) , \varepsilon\Bigr) \nonumber\\
\leq \varepsilon^{-2} \left\|\bigl[\mathbb{P}(z) -P\bigr]G(z) M' (z) G(z)
\mathbb{P}(z)\right\|_{HS}^2 \nonumber\\\leq 
\varepsilon^{-2} \Bigl\|[\mathbb{P}(z) -P]G(z) \Bigr\|^2
\left\|\mathfrak{X}_{[0, r_\varepsilon ]}M' (z) G(z)\right\|_{HS}^2  \label{09.74.aaa}
\end{gather}
Note that from (\ref{14}) it follows that $\|\varphi_\pm (z) - \eta_\pm \| \leq c z^2
|\ln z|$ for some $c >0$. Together with (\ref{9.61})-(\ref{34}) this gives 
\begin{equation}
 \varlimsup_{z \to 0} \bigl\|[\mathbb{P}(z) -P]G(z) \bigr\| \leq c'
r_\varepsilon , \label{09.74}
\end{equation}
where $c' >0$ is a constant. Now (\ref{09.71}) for $i=0$ follows from (\ref{09.74.aaa}), (\ref{09.74}) and Lemma~\ref{lem:2} if we choose
$r_\varepsilon$ small enough. Now let us prove (\ref{09.71}) for $i=1$. By (\ref{traceid}) and Lemma~\ref{lem:2} 
\begin{gather}
  \mathfrak{n}_\mu ( \mathcal{R}_1 (z), \varepsilon) \leq \varepsilon^{-2}
\left\|  \mathbb{Q}(z)\bigl(1 - M_d (z)\bigr)^{-\frac 12} \right\|^2 
\left\|\mathfrak{X}_{[0, r_\varepsilon ]} M' (z) G(z)\right\|_{HS}^2 \nonumber\\
\leq c_0 \varepsilon^{-2}r_\varepsilon^2 \left\|  \mathbb{Q}(z)\bigl(1 - M_d
(z)\bigr)^{-\frac 12} \right\|^2 |\ln |\ln z^2|| . 
\end{gather}
From (\ref{09.81}) it follows that 
\begin{equation}
\varlimsup_{z\to 0} \left\|  \mathbb{Q}(z)\bigl(1 - M_d (z)\bigr)^{-\frac 12}
\right\| < \infty 
\end{equation}
and thus (\ref{09.71}) holds true for $i=1$ if $r_\varepsilon$ is chosen small enough. The
proof of (\ref{09.71}) for $i=2$ is done analogously. 
\end{proof}
The proof of the next lemma is based on the following fact, which is proved on
page 40 in \cite{17}. Consider the operator $f(x)g(-i\nabla)$ acting on $L^2
(\mathbb{R}^4)$, where 
$f,g \in L^2_3$ (for notations see \cite{17}). Then 
\begin{equation}
 \|f(x)g(-i\nabla)\|_1 \leq C \left\{\int (1+|x|^2)^3 |f(x)|^2 d^4
x\right\}^{\frac 12}\left\{\int (1+|x|^2)^3 |g(x)|^2 d^4 x\right\}^{\frac 12} , \label{traceidreq}
\end{equation}
where the constant $C$ does not depend on $f,g$. 

\begin{proof}[Proof of Lemma~\ref{newlem:6}]
Let us first consider (\ref{09.72}) for the case when $i = 4$. On one hand, by (\ref{09.94}) 
\begin{equation}
\mathfrak{n}_\mu (\mathcal{R}_4 (z) , \varepsilon ) \leq  \mathfrak{n}_\mu \bigl( \mathfrak{X}_{(r_\varepsilon , \infty)} M'(z) , c^{-1} \varepsilon \bigr) , \label{09.106}
\end{equation}
where 
\begin{equation}
 c= \sup_{z\in (0, r_\varepsilon ]} \Bigl\|\bigl(1 - M_d (z)\bigr)^{-1/2} \mathfrak{X}_{(r_\varepsilon , \infty)} \Bigr\|^2  < \infty . 
\end{equation}
On the other hand there is a constant $c' >0$ such that 
\begin{equation}
 \bigl\|\mathfrak{X}_{(r_\varepsilon , \infty)} (M'(z) - M'(z')\bigr\|\leq c' |z^2 -z'^2|. \label{09.99}
\end{equation}
Eq. (\ref{09.99}) follows from (\ref{09.100}) after the applying the resolvent identity. From (\ref{09.99}) it follows that the operators $\mathfrak{X}_{(r_\varepsilon , \infty)} M'(z)$ form a 
Cauchy sequence for $z \to 0$ and converge in norm to a compact operator. Hence, the lhs of (\ref{09.106}) is bounded by a constant for $z\in (0, r_\varepsilon]$ and 
(\ref{09.72}) for $i = 4$ is proved. Now let us consider (\ref{09.72}) for $i=3$. By (\ref{25.08}) 
\begin{equation}
\mathfrak{n}_\mu (\mathcal{R}_3 (z), \varepsilon) \leq \mathfrak{n}_\mu
(\mathcal{R}_3^{(1)} (z), \varepsilon/3) + 
\mathfrak{n}_\mu (\mathcal{R}_3^{2)} (z), \varepsilon/3) + 
\mathfrak{n}_\mu (\mathcal{R}_3^{(3)} (z), \varepsilon/3), 
\end{equation}
where 
\begin{gather}
 \mathcal{R}_3^{(1)} (z) = \mathfrak{X}_{[0, r_\varepsilon]} \mathbb{P}(z)
\mathcal{A} (z) \mathfrak{X}_{(r_\varepsilon , R]} \nonumber \\
\mathcal{R}_3^{(2)} (z) = \mathfrak{X}_{[0, r_\varepsilon]} \mathbb{Q}(z)
\mathcal{A} (z) \mathfrak{X}_{(r_\varepsilon , \infty)} \nonumber \\
\mathcal{R}_3^{(3)} (z) = \mathfrak{X}_{[0, r_\varepsilon]}  \mathbb{P}(z) \mathcal{A} (z)
\mathfrak{X}_{(R , \infty)} \nonumber , 
\end{gather}
and $R \in (r_\varepsilon , \infty)$ is a parameter. Using the continuity arguments in the beginning of the proof one easily shows that 
\begin{equation}
\varlimsup_{z \to 0} |\ln|\ln z^2|^{-1} \mathfrak{n}_\mu (\mathcal{R}_3^{2)}
(z), \varepsilon/3) = 0
\end{equation}
for all values of $R$. From Lemma~\ref{lem:2} it follows that 
\begin{equation}
\varlimsup_{z \to 0} |\ln|\ln z^2|^{-1} \mathfrak{n}_\mu (\mathcal{R}_3^{3)}
(z), \varepsilon/3) = \hbox{o} (1/R) . 
\end{equation}
Thus instead of (\ref{09.72}) for $i=3$ it suffices to prove that 
\begin{equation}
\varlimsup_{z \to 0} |\ln|\ln z^2|^{-1} \mathfrak{n}_\mu (\mathcal{R}_3^{1)}
(z), \varepsilon/3) = 0 
\end{equation}
for all fixed $r_\varepsilon, \varepsilon, R >0$. Like in the proof of Lemma~\ref{lem:2} it suffices to show that 
\begin{equation}
\varlimsup_{z \to 0} |\ln|\ln z^2|^{-1} \mathfrak{n}_\mu (\mathcal{K} (z),
\varepsilon_0) = 0 \label{09.274}
\end{equation}
for all fixed $r_\varepsilon, \varepsilon_0 , R >0$, where 
\begin{equation}
\mathcal{K} (z) = v_1 \mathcal{F}_1^{-1} g_z (|q_1|) \mathcal{F}_1
\mathcal{F}_2^{-1} (p_2^2 +q_2^2 + z^2)^{-1} \chi_{(r_\varepsilon , R]} (|q_2|)
\mathcal{F}_2 v_2 
\end{equation}
acts on the space $L^2 (\mathbb{R}^4)$. Let us split $\mathcal{K} (z)$ into two
parts $\mathcal{K} (z) = \mathcal{K}_1 (z) + 
\mathcal{K}_2 (z)$, where $\mathcal{K}_1 (z) = \mathcal{K} (z) \chi_{[0, R']}
(|x_1|)$ and $\mathcal{K}_2 (z)  = \mathcal{K} (z) \chi_{(R', \infty)} (|x_1|)$
and $R' > 0$ is a parameter. By (\ref{25.08}) and (\ref{traceid})
\begin{equation}
\mathfrak{n}_\mu (\mathcal{K} (z), \varepsilon_0) \leq 2\varepsilon_0^{-1}
\|\mathcal{K}_1 (z)\|_1 + 4\varepsilon_0^{-2} \|\mathcal{K}_2 (z)\|^2_{HS} . \label{09.331}
\end{equation} 
Using formula (\ref{traceidreq}) we obtain the bound 
\begin{gather}
\|\mathcal{K}_1 (z)\|_1^2 \leq c \|v\|_\infty^2 \int d^2 x \int d^2 y (1+|x|^2 +|y|^2)^3
\chi_{[0, R']} (|x|) v^2 \left(\left| \frac 12 x - \frac{\sqrt 3}2 y\right|
\right) \nonumber\\
\times \int_{|p|\leq r_\varepsilon} \frac{d^2 p}{(p^2 + z^2)|\ln (p^2 + z^2)|} 
\int_{r_\varepsilon \leq \left|\frac{\sqrt 3}2 p + \frac 12 q \right| \leq R}
d^2 q 
\frac{(1+p^2 + q^2)^3 }{ (p^2 + q^2 +z^2)^2} \leq c(R') |\ln |\ln z^2|| , \label{09.111}
\end{gather}
where $c(R') \in (0, \infty)$ is a constant, which depends on $R'$. (The first
integral in (\ref{09.111}) converges because $v$ can be bounded by the exponent). 

Let us write the Fourier transform $\mathcal{F}_1$ as a product $\mathcal{F}_1 =
\mathcal{F}_x \mathcal{F}_y$, where $\mathcal{F}_{x,y}$ are partial 
Fourier transforms in variables $x_1$ and $y_1$ respectively. The operator 
$\mathcal{F}_y \mathcal{K}_2 (z) \mathcal{F}_y^{-1}$ can be written as a
product 
$\mathcal{F}_y \mathcal{K}_2 (z) \mathcal{F}_y^{-1} = \mathcal{K}_2^{(1)}
(z)\mathcal{K}_2^{(2)} (z)$, where 
\begin{gather}
 \mathcal{K}_2^{(1)} (z) = \chi_{[0, r_\varepsilon]}(|q_1|) v(|x_1|)\mathcal{F}_x^{-1}  \bigl(p_1^2 +q_1 ^2 + z^2 \bigr)^{-1} \nonumber\\
\times \chi_{[r_\varepsilon , R]} \left(\Bigl| \frac{\sqrt 3}2 p_1 + \frac 12 q_1\Bigr| \right)\mathcal{F}_x
\chi_{(R', \infty)} (|x_1|) \\
\mathcal{K}_2^{(2)} (z) =  g_z (|q_1|)  \mathcal{F}_y v\Bigl(-\frac 12 x_1 +
\frac{\sqrt 3}2 y_1\Bigr) \mathcal{F}^{-1}_y
\end{gather}
The integral operator $ \mathcal{K}_2^{(1)} (z)$ has the kernel $
\mathcal{K}_2^{(1)} (x_1, x_1' ; q_1 , z)$ and acts on $f(x_1 , q_1) \in L^2
(\mathbb{R}^4)$ as follows 
\begin{equation}
 [\mathcal{K}_2^{(1)} (z) f] (x_1 , q_1) = \int \mathcal{K}_2^{(1)} (x_1, x_1' ;
q_1 , z) f(x_1' , q_1) d^2 x_1' . 
\end{equation}
Similarly, $ \mathcal{K}_2^{(2)} (z)$ has the kernel 
\begin{equation}
 \mathcal{K}_2^{(2)} (q_1, q_1' ; x_1 , z) = \frac 2{3\pi} g_z (|q_1|) \exp
\Bigl(\frac{i}{\sqrt 3} (q_1-q_1')\cdot x_1\Bigr) \hat v \bigl((2/\sqrt 3 ) (q_1
- q_1')\bigr) , \label{09.221}
\end{equation}
where $\hat v$ is the Fourier transform of $v$. Now using (\ref{09.221}) we can estimate
the Hilbert-Schmidt norm 
\begin{equation}
  \|\mathcal{K}_2 (z)\|^2_{HS} = \|  \mathcal{K}_2^{(1)} (z) \mathcal{K}_2^{(2)}
(z) \|^2_{HS} \leq \beta d(R') |\ln |\ln z^2|| , \label{09.251}
\end{equation}
where $\beta >0$ is a fixed constant and 
\begin{equation}
 d(R') = \sup_{\substack{|q_1| \leq r_\varepsilon \\ z \in [0, r_\varepsilon ]}} \int
\bigl| \mathcal{K}_2^{(1)} (x_1, x_1' ; q_1 , z) \bigr|^2 d^2 x_1 d^2x_1' . 
\end{equation}
Let us show that $d(R') \to 0$ for $R' \to \infty$. Consider the compact integral operator $\mathcal{G}(q_1, z)$, which depends on the parameters 
$q_1, z$, acts on $L^2(\mathbb{R}^2)$ and has the structure $\mathcal{G} (q_1, z)= v(|x|)g(-i\nabla)$, where 
\begin{equation}
 g(s) = (s^2 + q_1^2 +z^2)^{-1}\chi_{[r_\varepsilon , R]} \left(\Bigl| \frac{\sqrt 3}2 s + \frac 12 q_1\Bigr| \right) . 
\end{equation}
Then it is easy to see that 
\begin{equation}
 d(R') = \sup_{\substack{|q_1| \leq r_\varepsilon \\ z \in [0, r_\varepsilon ]}} \bigl\|\mathcal{G} (q_1, z) \chi_{(R', \infty)} (|x|)\bigr\|^2_{HS} . \label{09.224}
\end{equation}
For fixed $q_1, z$ the operator $\mathcal{G} (q_1, z)$ is Hilbert-Schmidt, this can be checked by using Eq. (4.7) in \cite{17}. Thus for each fixed $q_1, z$ the expression  
under supremum in (\ref{09.224}) goes to zero for $R' \to \infty$. Using the same Hilbert-Schmidt norm estimate it is easy to show that for all $\epsilon >0$ there is $\eta >0$ 
such that $\|\mathcal{G} (q_1, z) - \mathcal{G} (q_1', z')\|_{HS} < \epsilon/2$  if $|q_1 - q_1'| < \eta$ and $|z-z'|< \eta$. We can cover the set 
$\{q_1|\, |q_1|\leq r_\varepsilon\} \cup \{z|\, z \in (0, r_\varepsilon]\}$ with the finite number of points $(q_1^{(i)}, z^{(i)})$ in such a way that for any 
$(q_1, z) \in \{q_1|\, |q_1|\leq r_\varepsilon\} \cup \{z|\, z \in (0, r_\varepsilon]\}$ 
there would exist  $(q_1^{(i_0)}, z^{(i_0 )})$ such that $|q_1 - q_1^{(i_0)}| < \eta$ and $|z-z^{(i_0 )}|< \eta$. Let us set $R'$ so that 
$\max_i \|\mathcal{G} (q_1^{(i)}, z^{(i)}) \chi_{(R', \infty)} (|x|)\bigr\|_{HS} <\epsilon $. Then we would have $d(R') < \epsilon^2$. Since $\epsilon$ is arbitrary 
we conclude that $d(R') \to 0$ for $R'\to \infty$. 

Summarizing, due to (\ref{09.331}), (\ref{09.111}) and (\ref{09.251}) we have the
upper bound
\begin{equation}
 \mathfrak{n}_\mu (\mathcal{K} (z), \varepsilon_0) \leq 2\varepsilon_0^{-1}
[c(R')]^{\frac 12} |\ln |\ln z^2||^{\frac 12} + 4\varepsilon_0^{-2} \beta
d(R')|\ln |\ln z^2|| . 
\end{equation}
Thus it follows that the lhs in (\ref{09.274}) is bounded by a fixed constant times $d(R')$.
Letting $R' \to \infty$ we complete the proof of (\ref{09.72}) for $i=3$. 
\end{proof}

\begin{acknowledgments}
The author would like to thank Artem Volosniev and the anonymous referees for
their critical reading of the manuscript and 
Jesper Levinsen for drawing attention to the reference \cite{jesper}. 
\end{acknowledgments}

\end{document}